%% file: 24-IS-DFF-PrivacyMechBoxes.tex
\documentclass[preprint,12pt,authoryear]{elsarticle}
\input{macros/preamble}

\journal{Information Science}

\begin{document}

\begin{frontmatter}
\input{macros/authors}
\input{macros/abstract}

\end{frontmatter}


\input{sections/01introduction}
\input{sections/02theoreticalBackground}

\input{sections/03methodology}
\input{sections/04experimentalAnalysis}

\input{sections/05relatedWorks}
\input{sections/06conclusionsAndFutureWorks}




\input{references}

\end{document}

%% file: macros/preamble.tex
\usepackage[linesnumbered,ruled,vlined]{algorithm2e}
\usepackage{multirow}
\usepackage{xcolor}
\usepackage{soul}
\usepackage{acronym}
\usepackage{subfigure}
\usepackage{braket}
\usepackage{amssymb}
\usepackage{amsmath}
\usepackage{amsthm}
\usepackage{url}
\usepackage{booktabs} 

\newtheorem{theorem}{Theorem}

\input{macros/acronyms}

\SetCommentSty{mycommfont}

\input{macros/local}
\SetKwInput{KwInput}{Input}                
\SetKwInput{KwOutput}{Output}              

%% file: macros/acronyms.tex
\acrodef{AD}[AD]{Aggregated Data}
\acrodef{ANOVA}[ANOVA]{ANalysis Of the VAriance}
\acrodef{AP}[AP]{Average Precision}
\acrodef{BEIR}[BEIR]{(BEnchmarking IR}
\acrodef{BoW}{Bag-of-Words}
\acrodef{CS}[CS]{Cosine Similarity}
\acrodef{CMP}[CMP]{Calibrated Multivariate Perturbation}
\acrodef{DC}[DC]{Dense-Centroid}
\acrodef{DL}[DL]{Deep Learning}
\acrodef{DP}[DP]{Differential Privacy}
\acrodef{DQPP}[DQPP]{DenseQPP}
\acrodef{DRMM}[DRMM]{Deep Relevance Matching Model}
\acrodef{HSD}[HSD]{Honestly Significant Differences}
\acrodef{HV}[HV]{Hypervolume}
\acrodef{ICTF}[ICTF]{Inverse Collection Term Frequency}
\acrodef{IDF}[IDF]{Inverse Document Frequency}
\acrodef{IPD}[IPD]{Individual Participant Data}
\acrodef{IR}[IR]{Information Retrieval}
\acrodef{IRS}[IRS]{Information Retrieval System}
\acrodef{MEM-QPP}[MEM-QPP]{MEMory-based QPP}
\acrodef{MHL}[Mhl]{Mahalanobis}
\acrodef{JS}[JS]{Jaccard Similarity}
\acrodef{KL}[KL]{Kullback–Leibler}
\acrodef{LM}[LM]{Language Model}
\acrodef{LLM}[LLM]{Large Language Model}
\acrodef{LS}[LS]{Lexical Similarity}
\acrodef{LeToR}[LeToR]{Learning-To-Rank}
\acrodef{LoTTE}[LoTTE]{Long-Tail Topic-stratified Evaluation}
\acrodef{ML}[ML]{Machine Learning}
\acrodef{MRR}[MRR]{Mean Reciprocal Rank}
\acrodef{nDCG}[nDCG]{normalized Discounted Cumulative Gain}
\acrodef{NIR}[NIR]{Neural Information Retrieval}
\acrodef{NLP}[NLP]{Natural Language Processing}
\acrodef{NQC}[NQC]{Normalized Query Commitment}
\acrodef{P}[P]{Precision}
\acrodef{PCA}[PCA]{Principal Component Analysis}
\acrodef{PDQPP}[PDQPP]{Projection Displacement Query Performance Predictor}
\acrodef{PLM}[PLM]{Pre-trained Language Model}
\acrodef{POS}[POS]{Part-Of-Speech}
\acrodef{PRF}[PRF]{Pseudo-Relevance Feedback}
\acrodef{QA}[QA]{Question Answering}
\acrodef{QPP}[QPP]{Query Performance Prediction}
\acrodef{RBO}[RBO]{Rank-Biased Overlap}
\acrodef{sARE}[sARE]{scaled Absolute Rank Error}
\acrodef{SE}[SE]{Search Engine}
\acrodef{SemS}[SemS]{Semantic Similarity}
\acrodef{SCNQC}[SCNQC]{Scaled-Combined Normalized Query Commitment}
\acrodef{SCS}[SCS]{Simplified query Clarity Score}
\acrodef{SCQ}[SCQ]{Similarity Collection-Query}
\acrodef{sMARE}[sMARE]{scaled Mean Absolute Rank Error}
\acrodef{SMV}[SMV]{Score Magnitude and Variance}
\acrodef{SV}[SV]{Score Variance}
\acrodef{TIR}[TIR]{Traditional Information Retrieval}
\acrodef{UEF}[UEF]{Utility Estimation Framework}
\acrodef{US}[US]{Uncertainty Statistics}
\acrodef{WBB}[WBB]{Words Blending Boxes}
\acrodef{WIG}[WIG]{Weighted Information Gain}
\acrodef{WRIG}[WRIG]{Weighted Relative Information Gain-based model}

%% file: macros/local.tex
\definecolor{ggreen}{rgb}{0.061, 0.563, 0.251}

\definecolor{vviolet}{RGB}{160, 51, 255}

\newcommand{\argmax}{\operatornamewithlimits{\sf argmax}}

\newcommand{\mycaption}[1]{\caption{#1}}

\newcommand{\fea}{FEA\xspace}
\newcommand{\aea}{AEA\xspace}

%% file: macros/authors.tex


\title{Words Blending Boxes. Obfuscating Queries in Information Retrieval using Differential Privacy.}


\author[1]{Francesco Luigi De\ Faveri\corref{cor1}}\ead{defaverifr@dei.unipd.it}\ead[url]{https://www.dei.unipd.it/~defaverifr/}
\author[1]{Guglielmo Faggioli\corref{cor1}}\ead{faggioli@dei.unipd.it}\ead[url]{https://www.dei.unipd.it/~faggioli/}
\author[1]{Nicola Ferro\corref{cor1}}\ead{ferro@dei.unipd.it}\ead[url]{https://www.dei.unipd.it/~ferro/}
\cortext[cor1]{Corresponding Author(s).}

\affiliation[1]{organization={Department of Information Engineering, University of Padua},
            city={Padua},
            postcode={35131}, 
            country={Italy}}

%% file: macros/abstract.tex
\begin{abstract}
Ensuring the effectiveness of search queries while protecting user privacy remains an open issue. When an \acf{IRS} does not protect the privacy of its users, sensitive information may be disclosed through the queries sent to the system. Recent improvements, especially in \ac{NLP}, have shown the potential of using \acf{DP} to obfuscate texts while maintaining satisfactory effectiveness. By perturbing the non-contextualized vector representations of the terms within the query, terms are replaced with the one closest to the perturbed version. However, such approaches may protect the user's privacy only from a theoretical perspective while, in practice, the real user's information need can still be inferred if perturbed terms are too semantically similar to the original ones. Furthermore, a poor tuning of the \ac{DP} privacy budget $\varepsilon$, i.e., choosing a too high value, is bound to cause catastrophic data release since, in such case, the user's query will be released as if no privacy protection was adopted at all. We overcome such limitations by proposing \acf{WBB}, a novel \ac{DP} mechanism for query obfuscation, which protects the words in the user queries by employing safe boxes. On the one side, these boxes avoid too semantically similar words to be sampled; on the other side, they rely on a consolidated \ac{DP} mechanism for sampling from a set of candidate words to mask the original query. 
To measure the overall effectiveness of the proposed \ac{WBB} mechanism, we measure the privacy obtained by the obfuscation process, i.e., the lexical and semantic similarity between original and obfuscated queries. Moreover, we assess the effectiveness of the privatized queries in retrieving relevant documents from the \ac{IRS}.
Our findings indicate that \ac{WBB} can be integrated effectively into existing \ac{IRS}, offering a key to the challenge of protecting user privacy from both a theoretical and a practical point of view. 
\end{abstract}



\begin{keyword}
Differential Privacy\sep Information Hiding\sep Information Retrieval\sep Information Security


\end{keyword}

%% file: sections/01introduction.tex
\section{Introduction}
\label{introduction}
\acfp{IRS} are a daily commodity used to satisfy the most heterogeneous user information needs. One of the most well-known examples of \ac{IRS} are Search Engines on which the user issues a natural language query and the system retrieves a set of web pages that should contain the answer to the user query.  
One of the risks linked to \ac{IRS} is that, when searching, the queries issued by the user could threaten their \emph{privacy}. An adversarial, including the \ac{IRS} itself, can leverage the users' queries to infer what they are searching for, gaining access to sensitive information about them. Imagine a user conducting multiple searches for documents related to a specific medical condition~\citep{kumok2020covidsearch,zimmerman2030TowardsSearchStrategiesforBetterPrivacyandInformation}. In this case, a malicious employee of the \ac{IRS} accessing the query logs could infer sensitive information regarding the user's health. Another example is \emph{ego surfing}, a common practice where users search for their name, social security number, or social profile, potentially allowing the \ac{IRS} to link the request with a real person. Recent studies have shown how privacy violations in analyzing search histories can leak the political views~\citep{le2019googlepolitical,mustafaraj2020politicalbias} or the sexual orientation~\citep{malandrino2013privacy,shekhawat2019genderbiasinads} of a user. Such knowledge can lead to unfair treatment and risks in illiberal countries.

To address this, we would like to allow the user to send the \ac{IRS} queries that are not sensitive, but still allow the user to retrieve relevant documents as if the original sensitive query was sent. This would prevent the adversarial (i.e., the \ac{IRS}) from profiling the user.
This threat model requires us to operate under the \emph{non-collaborative} \ac{IRS} assumption: we consider the \ac{IRS} as not interested in protecting the user's privacy, but rather as a potential threat. 
For obvious reasons, encryption, in this case, is not a suitable protection mechanism. The \ac{IRS} is supposed to answer the user's query thus being it is its intended receiver. While encryption can be effectively adopted to protect the query from external eavesdroppers -- assuming the \ac{IRS} is willing to agree on an encryption schema --, it does not provide any guarantee concerning the information leakage that might occur directly on the \ac{IRS} side.
Secondly, under the hypothesis that the \ac{IRS} is non-cooperative, we might not be able to agree upon an encryption protocol for communicating with the \ac{IRS}.


One possibility to mitigate such a threat involves applying text obfuscation mechanisms.
Broadly speaking, an obfuscation mechanism takes in input a sensitive piece of text $t$ and outputs a new piece of obfuscated text $t'$, that is semantically related to $t$, but not sensitive.

The text obfuscation task has been extensively investigated in the \acf{NLP} domain. 
The most common scenarios in which text obfuscation is used in \ac{NLP} involve tasks such as sentiment analysis, spam detection, and text classification~\citep{zhao2022survey,feyisetan2019leveraging,feyisetan2020multivariate,feyisetan2021private,xu2020mahalanobis,XuAggarwalEtAl2021}.
In \ac{NLP}, the most common framework to obfuscate and release text is \acf{DP}~\citep{dwork2006calibrating}.
State-of-the-art solutions based on \ac{DP} for text obfuscation involve perturbing each word of a piece of text~\citep{zhao2022survey,feyisetan2019leveraging,feyisetan2020multivariate,feyisetan2021private,xu2020mahalanobis,XuAggarwalEtAl2021}. More in detail, given a non-contextualized embedding of each word in the text, such embedding is perturbed by adding appropriately sampled noise. Then, a new obfuscated text is constructed by taking, for each perturbed word embedding, the closest word to it. In this way, the original content of the text is safeguarded. 

Nevertheless, when it comes to \ac{IRS}, obfuscation approaches designed for \ac{NLP} do not necessarily work smoothly, for two main reasons.
First, when it comes to \ac{NLP}, the obfuscated text can be used both at training and inference time: this allows a machine learning mechanism (e.g., a classifier) to account for the noise introduced in the text. The same cannot occur in our scenario. Indeed, we assume the \ac{IRS} to be non-cooperative: it would be impossible to operate on its training to account for the obfuscated queries.
Secondly, one of the challenges linked to \ac{DP} is that too-small noise can lead to catastrophic data releases in which the text is not obfuscated. In other terms, tuning the privacy guarantees (i.e., modifying the privacy budget $\varepsilon$) of a DP mechanism provides privacy formally but, in practice, it does the obfuscated text. This is a well-known and severe problem of \ac{DP}~\citep{DomingoferrerSanchezBlancojusticia2021}. Indeed, none of state-of-the-art solutions proposed for text obfuscation in \ac{NLP}~\citep{zhao2022survey,feyisetan2019leveraging,feyisetan2020multivariate,feyisetan2021private,xu2020mahalanobis,XuAggarwalEtAl2021}, fully overcomes the problem, as they still allow the original word to be picked as the obfuscation term and highly semantically similar terms.

Besides generic approaches to deal with \ac{NLP} at large, some efforts have been specifically targeted to the query obfuscation task, in which the text to be obfuscated is a query that needs to be sent to a \ac{IRS}. 
The first of such approaches, proposed by \citet{arampatzis2013queriesscrambler}, replaces words in queries with randomly sampled hypernyms of them, under the assumption that more general words are less sensitive. On the contrary, the approach proposed by \citet{frobe2022efficentqueryobf} employs a local corpus to determine which terms co-occur with the query terms in the documents the most frequently. Using this information, they construct obfuscation queries that only contain terms that often co-occur with query terms, so that both the original and obfuscated queries will retrieve similar documents.
Nevertheless, by relying on strictly semantic relations between words and terms co-occurrences, current approaches might release terms that include words highly semantically related to or synonyms of those originally contained in the queries.
Furthermore, no approach to query obfuscation for \ac{IRS} is based on \ac{DP}, thus making it impossible to grant demonstrable privacy to the obfuscation queries produced. 
More recently, \cite{FaggioliFerro2024} demonstrated the effectiveness of \ac{NLP}-inspired \ac{DP} techniques in the \ac{IR} domain. Nevertheless, they do not propose a new approach that would be more effective in the \ac{IR} domain. 

To overcome all the aforementioned limitations, we propose \acf{WBB}, a novel approach to safeguard user privacy by obfuscating the words in the original query, creating a safe box around the query terms in the non-contextual embedding space, which ensures that the query words and terms too similar to the original ones will not be present in the final obfuscated output. Furthermore, we ensure that the outcome of our approach is private by using \ac{DP} to select which terms to use for the obfuscated query among those outside the safe box. This guarantees that our proposed solution not only satisfies the privacy requirements from a theoretical point of view but, thanks to the safe box, also provides practical privacy for the user. To the best of our knowledge, \ac{WBB} represents the first \ac{DP} approach explicitly designed for the query obfuscation task in the \ac{IR} context.

By conducting extensive experimentation on both classical (Robust '04) and neural-oriented (Deep Learning '19) TREC collections, the \ac{WBB} method is evaluated and compared with current State-of-the-Art mechanisms in three aspects: privacy, recall, and utility.
\begin{itemize}
    \item Firstly, we determine if our protection mechanism can provide sufficient \textbf{privacy guarantees}, beyond the formal privacy granted by using \ac{DP}. To do so, we compute the lexical and semantic similarity between the original and obfuscated queries, showing that our mechanism induces sufficiently obfuscated queries.

    \item Secondly, one of the risks linked to obfuscation mechanisms, especially when high privacy is enforced, is that obfuscated queries might not be able to \textbf{retrieve relevant documents}. Therefore, we are interested in determining if the obfuscation queries can retrieve relevant content. To do so, we measure the recall observed when the documents retrieved by different obfuscation queries are combined, that is what the user would do upon receiving the \ac{IRS} answers to the obfuscation queries. 

    \item Finally, we are interested in \textbf{determining the utility preserved} by our obfuscation mechanism. To do so, we assume the user reranks the documents received by the \ac{IRS} using the original query. This is a safe operation as it occurs on the user side. As a proxy of the experienced utility, we compute the nDCG@10 for the reranked list of documents.
\end{itemize}


The paper is organized as follows: Section~\ref{sec:theoreticalBackground} introduces the preliminaries on \ac{DP}, along with important \ac{DP} mechanisms that the \ac{WBB} uses. Moreover, Section~\ref{sec:background-on-text-obfuscation} describes the \ac{NLP} method used for text obfuscation, while Section~\ref{sec:methodology} explains our methodology; Section~\ref{sec:experimentalAnalysis} reports the results of our experiments; finally, Section~\ref{sec:relatedWorks} presents related works, and Section~\ref{sec:conclusionsAndFutureWorks} draws some conclusions and outlooks for future work. 

%% file: sections/02theoreticalBackground.tex
\section{Background on Differential Privacy}
\label{sec:theoreticalBackground}
This section reports the theoretical foundations of the \ac{DP} framework, underlying the mechanism designed in this work.

\subsection{Privacy Definitions}
\acf{DP}, introduced by~\citet{dwork2006calibrating}, is, de facto, the gold-standard formal definition used to determine if an algorithm protects privacy. Intuitively, a \ac{DP} mechanism adds, during the computation, an appositely crafted amount of noise that depends on the privacy budget $\varepsilon$, which sets the trade-off between data privacy and utility.

The definition of $\varepsilon$-\ac{DP}~\citep{dwork2006calibrating} states that a randomized mechanism $\mathcal{M}$ (i.e., an algorithm that takes a certain input and produces a noisy output) is $\varepsilon$-\ac{DP} if, for any pair of neighbouring datasets $D$ and $D'$, i.e., two datasets differing of a single record, and a privacy budget $\varepsilon\in\mathbb{R}^{+}$, it holds:
\begin{equation}
    \text{Pr}\{\mathcal{M}(D)\in \mathcal{S}\}\leq e^{\varepsilon}\text{Pr}\{\mathcal{M}(D')\in \mathcal{S}\},\ \forall \mathcal{S}\subseteq \text{Image}(\mathcal{M})
\label{eq:DP}
\end{equation}

If the mechanism is $\varepsilon$-\ac{DP}, the definition grants that for every run of the randomized mechanism $\mathcal{M}$, the output (i.e., a value from $\mathcal{S}$) is almost equally likely to be observed on every neighbouring dataset simultaneously. By definition, lower values of $\varepsilon$ guarantee higher levels of privacy: if $\varepsilon=0$, then $\text{Pr}\{\mathcal{M}(D)\in \mathcal{S}\} = \text{Pr}\{\mathcal{M}(D')\in \mathcal{S}\}\ \forall \mathcal{S}\subseteq \text{Image}(\mathcal{M})$, i.e., the output does not depend on the input.
Intuitively, given two similar yet different inputs, we expect the output to be the same with a certain probability regulated by $\varepsilon$. This grants ``Plausible deniability'': the adversarial cannot deem with absolute certainty which input (i.e., user's data) corresponds to a given output.

\subsection{The Exponential Mechanism}
The mechanism used to sample words underlying \ac{WBB} builds upon the $\varepsilon$-\ac{DP} exponential mechanism~\citep{mcsherry2007exponential, dwork2014privacybook}. The general idea underneath the exponential mechanism is that given a certain input $x\in\mathcal{I}$, a range of interest $\mathcal{R}$, and a utility function $u$ such that $u:\mathcal{I}\times\mathcal{R}\to\mathbb{R}$, we are interested in outputting a value $r\in\mathcal{R}$ that, given the input $x$, maximizes the utility while remaining \ac{DP}. 

More in detail, to implement the exponential mechanism, it is first necessary to compute its \textit{sensitiveness}, which corresponds to the maximum difference in utility that can be observed for any possible value of the range $\mathcal{R}$ and any possible pair of neighbouring inputs (i.e., highly similar input data). The sensitivity of a given utility function $u$ is computed as:
\begin{equation}
    \Delta u = \max_{r\in \mathcal{R}}\max_{x,y:\|x-y\|_{1}\leq 1} |u(x,r)-u(y,r)|
\label{eq:utiltydef}
\end{equation}

The exponential mechanism outputs a value $r\in\mathcal{R}$ with probability proportional to $\exp\left(\frac{\varepsilon u(x,r)}{\Delta u}\right)$. This ensures that the output depends on the utility of $r$ given $x$ and the $\varepsilon$ used: a small $\varepsilon$ ensures that the utility plays a less and less prominent role in the sampling.
In practical terms, when instantiated in our scenario, given an input word, we sample another word based on some utility measure (e.g., the similarity with the original term). We are not guaranteed the most useful word is sampled, i.e., the most similar, as this would expose the input, but we are likely to sample a high utility term.

\section{Background on Text Obfuscation}
\label{sec:background-on-text-obfuscation}

In this section, we describe existing background and approaches to text obfuscation. We start from approaches designed for \ac{NLP} tasks, where it is the state-of-the-art to employ \ac{DP}. Then, we move to the specific \ac{IR} problem, where we observe that the approaches to obfuscate the queries are based on statistics of their terms.
Finally, we detail some approaches to measure the privacy provided by a given obfuscation approach.

\subsection{Text Obfuscation in NLP}
\label{subsec:DP-mechanisms}

The general task of a text obfuscation mechanism consists of taking in input sensitive piece of text $t$ and outputting a new piece of obfuscated text $t'$, that is semantically related to $t$, but not sensitive.
At this point, being $t'$ safe from the privacy perspective, it can be released without risking the user's privacy. More in detail, such a piece of text can be used for example to train a machine learning model in some downstream task. For example, $t$ might be a post written by a user on a social network. Assume now the objective is to train a sentiment classifier. To avoid risking the users' privacy, the text of the post can be first obfuscated and then passed as a training example to the classifier.
For most tasks, a single obfuscated text might not be sufficiently rich to convey the complexity of $t$, thus it is common to generate several pieces of non-sensitive text $t'_1$, $t'_2$, ..., $t'_n$.

Most recent state-of-the-art approaches for text obfuscation in \ac{NLP} at large employ~\ac{DP}~\citep{feyisetan2020multivariate,chen2023customized,yue2021santext,xu2020mahalanobis,carvalho2023tem,XuAggarwalEtAl2021}.
Most of these mechanisms adopt a similar strategy to obfuscate the text. Assume a piece of text $\mathcal{W}^l$ of length $l$ needs to be obfuscated. Let $\mathcal{V}$ be the vocabulary and $\phi$ be a non-contextualized word encoder such as Word2Vec~\citep{mikolov2013word2vec} or GloVe~\citep{pennington2014glove}. This word encoder takes in input a word $w$ and generates a non-contextual embedding representation $\phi(w)\in\mathbb{R}^d$, where $d$ is the dimension of the word embedding.
Let us call $\nu(\varepsilon)$ the noise appositely crafted accordingly to the noise distribution function $\nu$ to satisfy the \ac{DP} constraint. 

Then, given a word $w$, it is obfuscated with a word $w'$ as follows:
\begin{equation}
    \argmax_{w'\in \mathcal{V}} ||\phi(w') - (\phi(w)+\nu(\varepsilon))||
\end{equation}
Where $||\cdot||$ represents any norm (e.g., the Euclidean norm).
In other terms, the word $w'$ used to obfuscate $w$ is the one that has the representation closest to the non-contextual representation of $w$ perturbated with noise $\nu(\varepsilon)$.
These mechanisms were originally designed to operate in an \ac{NLP} context and explicitly for tasks such as sentiment analysis and spam detection. Notice that, a major difference between \ac{NLP} and \ac{IRS}, is that in the former case, it is possible to fine-tune the model on some obfuscated text. This ensures that the model latently learns how to account for the shift in the language model operated by the obfuscation mechanism.
In the case of \ac{IRS}, we cannot apply the same reasoning: under the hypothesis of a non-cooperative \ac{IRS}, we have no access to the model it relies upon, and we cannot retrain it externally. 

We now provide examples of works that explicitly employ~\ac{DP} in the \ac{NLP} domain.
\citet{zhao2022survey} provided a comprehensive overview of DP strategies literature to protect sensitive information in unstructured data. \citeauthor{zhao2022survey} highlighted the prominence of the trade-off between privacy and utility in privatizing sensitive textual information. Such privacy \emph{versus} utility trade-off originates because enhancing privacy involves altering the original data, which reduces the specificity and accuracy of the data. Therefore, the challenge is finding an optimal balance where privacy is sufficiently guaranteed without significantly compromising the privatized data utility. Hence, the authors conclude that \ac{DP} has a significant, unexplored potential to improve data utility while guaranteeing high privacy protection levels for unstructured data such as texts.

\citet{feyisetan2019leveraging, feyisetan2020multivariate, feyisetan2021private} proposed \ac{DP} methods based on the noise addition to word embeddings. Such methods depend on defining \emph{metric} \ac{DP} to perturb vectors. Metric-DP is a relaxation of the original definition introduced by \citet{Chatzikokolakis2013broadening}, which considers the metric function introduced in the vector space used for the word embeddings. 
The general idea is to take the non-contextualized embedding of the terms in the text, inject some noise, and consider the word whose embedding is the closest to it.
\citet{feyisetan2020multivariate} explored how to balance the trade-off between privacy and utility using calibrated noise and metric DP by leveraging geometric properties of word embeddings and practically evaluating the utility of the resulting scrambled texts. The study explores the impact that calibrated noise can have on helping users preserve their privacy in text obfuscation in NLP tasks. Yet, in such studies, privacy is measured without considering concretely how the mechanism carries on the obfuscation and how the obfuscated word is related to the original one.


\citet{xu2020mahalanobis} started from the observation that the mechanism proposed by \citet{feyisetan2020multivariate} is likely to obfuscate a word with itself. To account for this, propose a mechanism for DP text perturbation that calibrates the noise injection by incorporating the regularized Mahalanobis metric. On the same line, \citet{XuAggarwalEtAl2021} further extended~\cite{feyisetan2020multivariate} by proposing a mechanism based on the Vikrey mechanism, where a word can be obfuscated with either the most similar or the second most similar approaches.
In addition, other methods proposed for achieving good levels of text obfuscation using DP have been studied by \citet{yue2021santext, chen2023customized}. On the one hand, \citet{yue2021santext} implemented a \ac{DP} mechanism, \textsc{SanText}, to sanitize texts to achieve good levels of utility while still protecting the defined sensitive words in the text. On the other hand, \citet{chen2023customized} presented \textsc{CusText}, a privatization mechanism able to adapt with any similarity measure to balance the privacy-utility trade-off. While going in the direction of overcoming the limitation linked to a word being obfuscated by itself, this possibility remains present in both mechanisms, thus keeping the user exposed to the risk of catastrophic data release.

On a different line, \citet{carvalho2023tem} suggested the Truncated Exponential Mechanism to ensure user privacy in NLP model training. The algorithm utilizes the exponential mechanism~\cite{mcsherry2007exponential} to change the privatization procedure into a selection problem, enabling noise calibration based on embedding space density around a given input. Nevertheless, the limitation of the mechanism originates from its inability to include the privatization of word context vectors, consequently constraining the guaranteed privacy.

\subsection{Query Obfuscation in IR}

We introduce in this section the generic query obfuscation pipeline and describe how it has been instantiated in practice.

\subsubsection{The query obfuscation pipeline}

\input{figures/pipeline}

Figure~\ref{fig:pipelinemoverview} illustrates the general query obfuscation pipeline analyzed in this work and commonly studied in \ac{IR} scenarios implementing query obfuscation protocols. 
The process occurs at two distinct places: on the user side (\emph{safe}), where the user writes their private query and obfuscates it, and on the \ac{IRS} side (\emph{unsafe}), where the \ac{IRS} retrieves the documents given the obfuscated queries received. The system is oblivious to the real information need. 

To define a query obfuscation mechanism, we take inspiration from the \ac{NLP} case. Therefore, in the \ac{IRS} scenario, given a sensitive query $q$ (``Private Query in Figure~\ref{fig:pipelinemoverview}), we apply an obfuscation mechanism (``Obfuscation mechanism'' in Figure~\ref{fig:pipelinemoverview}) to generate several non-sensitive queries $q'_1$, $q'_2$, ..., $q'_n$ (``Obfuscated queries'' in Figure~\ref{fig:pipelinemoverview}). This occurs on the user side therefore it is safe from the privacy perspective.
Once the obfuscated queries have been generated, they can be safely transmitted to the \ac{IRS}.
These queries, are only semantically related to $q$ but focus on different topics. Therefore, it is likely that they will not retrieve all and every document relevant to $q$, and most likely not on the top positions. Therefore, the user transmits the $n$ obfuscation queries $q'_{\dots}$ to the \ac{IRS}. Upon receiving these queries, the \ac{IRS} computes a ranked list of documents for each of them (i.e., ``Retrieved docs'' in Figure~\ref{fig:pipelinemoverview}). Notice that this does not require the \ac{IRS} to be aware of the obfuscation mechanism being in place: the \ac{IRS} behaves as a black box that takes in input a query and outputs a list of documents.  
The user then receives a set of lists of documents. As it is likely that relevant documents are not in the first positions, post-filtering and re-ranking are performed using the original user's query. Finally, the user obtains a new ranked list with the likely relevant documents placed in the first positions.

We want to stress that, in this case, cryptographic protocols do not provide privacy for the user. The non-cooperative \ac{IRS} is the intended receiver of the query; hence it should be the only one supposed to decipher the cryptographic message. However, in the scenario proposed, the \ac{IRS} is interested in profiling the user and thus, cryptography is not enough to achieve privacy.
To protect the query from third-party eavesdroppers, safe communication protocols based on private or public key encryption, e.g., AES~\citep{rijndel2002AES} or RSA~\citep{rivest1978RSA}, can be employed. Nevertheless, these methods do not ensure privacy but only the confidentiality of communication between the user and the \ac{IRS}.


\subsubsection{State-of-the-art Approches}
Privacy via query obfuscation in IR without considering the definition of \ac{DP} has been studied in prior works. \citet{arampatzis2013queriesscrambler} proposed an obfuscation mechanism based on Word-Net~\citep{miller1995wordnet} in which each term in the query is generalized, implementing a hierarchical definition using the extracted synonyms, hypernyms, and holonyms of the word. The candidate obfuscation queries are then computed considering the Cartesian product of the terms sets and filtered out considering the similarity between the original and obfuscated queries to avoid exposing queries. An extension of the studies of \citet{arampatzis2013versatiletoolprivacywebsearch,arampatzis2015versativescramblerprivatewebsearch} have been introduced by~\citet{frobe2022efficentqueryobf} that developed a statistical query obfuscation method using the user's query on the local corpus. 

\citet{FaggioliFerro2024} bridged the gap between the two aforementioned research areas by showing how \ac{DP} mechanisms borrowed from \ac{NLP} may perform for query obfuscation in \ac{IR}. Nevertheless, \citeauthor{FaggioliFerro2024} directly applied such state-of-the-art \ac{DP} mechanisms devised for \ac{NLP}~\citep{XuAggarwalEtAl2021,xu2020mahalanobis,feyisetan2020multivariate}, falling short in devising a \ac{DP} approach that does not risk obfuscating a query with itself.

Therefore, to the best of our knowledge, this is the first effort to design and formally prove an $\varepsilon$-\ac{DP} method directly for the query obfuscation task in \ac{IR}.

\subsection{Privacy Measures}
\label{subsec:Privacy-metrics}

One of the major challenges when it comes to dealing with privacy and text concerns the evaluation of how much information is leaked by the obfuscated text compared to the original one. Measuring if two pieces of text convey the same meaning is a very challenging task: two sentences might be syntactically different but semantically equivalent. Vice versa, small changes, even of a few characters, might completely change the meaning of a piece of text. Additionally, it is hard to establish if something is sensitive and leaks private information. Often, the degree of sensitivity depends on the context: we might not feel comfortable sharing the very same information with a stranger as we would tell our doctor.

Therefore, frequently employed metrics in the fields of both \ac{NLP} and information security to assess the level of privacy are typically categorized based on the specific aspect of privacy they are measuring~\citep{wagner2018privacymetricssurvey}: Uncertainty Statistics, Lexical Similarity, Semantic Similarity.

\textit{Uncertainty Statistics} measure the probability of an obfuscation approach failing at the level of the single term. The two most adopted measures are $N_w$ and $S_w$. More in detail:
\begin{itemize}
    \item $N_{w} = \text{Pr}\left[\mathcal{M}(w)=w\right]$. It represents the probability that the mechanism obfuscates the word with itself (i.e., it fails the obfuscation). 
    \item $S_{w} = \min\left|\left\{\mathcal{S}\subseteq \text{Image}(\mathcal{M}) : \text{Pr}\left[\mathcal{M}(w)\notin \mathcal{S}\right]\leq\eta\right\}\right|$. It corresponds to the minimum size of the set that contains the words used to obfuscate the same input word.
 \end{itemize}
While these measures are simple to compute, they are too fine-grained to provide real insight into whether a piece of text has been correctly obfuscated: not all words are equally sensitive.

\textit{Lexical Similarity} measures try to overcome the limitations of uncertainty statistics, by measuring the quality of a piece of obfuscated text at the global level. Examples of such measures include approaches drawn from the automatic translation (with the caveat that a high lexical similarity indicates low obfuscation and -- likely -- low privacy). Examples of such measures include the Jaccard similarity, BLEU~\cite{PapineniRoukosEtAl2002}, or METEOR~\cite{LavieAgarwal2007}. In most cases, these measures involve computing the overlapping between the original and obfuscated text -- or their n-grams.
To provide an example, the Jaccard similarity is defined as follows:
$$\text{Jaccard Similarity}(A, B)=\dfrac{|A\cap B|}{|A \cup B|},$$
this measure is computed over the set of words in the original query $A$ and the obfuscated ones $B$, measuring the proportion of overlapping terms.
Even though they represent an improvement over uncertainty statistics with a global view of the text, lexical similarity approaches still lack sufficient semantic understanding to fully grasp the complexity underlying privacy-related tasks.

\textit{Semantic Similarity} measures employ a notion of semantic similarity to determine if a piece of text is an effective obfuscation of another one: too similar pieces of text do not guarantee privacy and to different texts do not provide utility. 
Automatically embedding the semantics of a text is a complex task, typically addressed by encoding the text in a latent space. The approach current state-of-the-art solutions for this task, such as BERT~\cite{DevlinChangEtAl2019}, involve the usage of the transformers architecture~\cite{VaswaniShazeerEtAl2017}.
In this case, the simplest strategy to compute the similarity between two texts involves computing their embedding $t$ and $t'$ in a latent space employing a suitable encoder (e.g., BERT). Then, the semantic similarity is approximated by the cosine similarity between the two vectors:
$$\text{Semantic Similarity}(t,t')=\dfrac{t\cdot t'}{\lVert t\rVert\lVert t'\rVert}.$$
While this approach can grasp the semantic similarity between two sentences, it is tightly linked to the encoder used to compute the embeddings. This means that 
i) different results might be observed depending on the encoder at hand; ii) if the obfuscation model relies on the same encoder used for the evaluation, there might be some leakage.

In this paper, we measure the similarity between the original and obfuscated queries in terms of Jaccard Similarity and Semantic similarity, to provide both lexical and semantic points of view. As mentioned before, a high similarity is an indication of low obfuscation (i.e., the original and obfuscated queries are too similar), thus weak privacy guarantees.

    


%% file: figures/pipeline.tex
\begin{figure*}[h!]
    \centering
    \includegraphics[width=\linewidth]{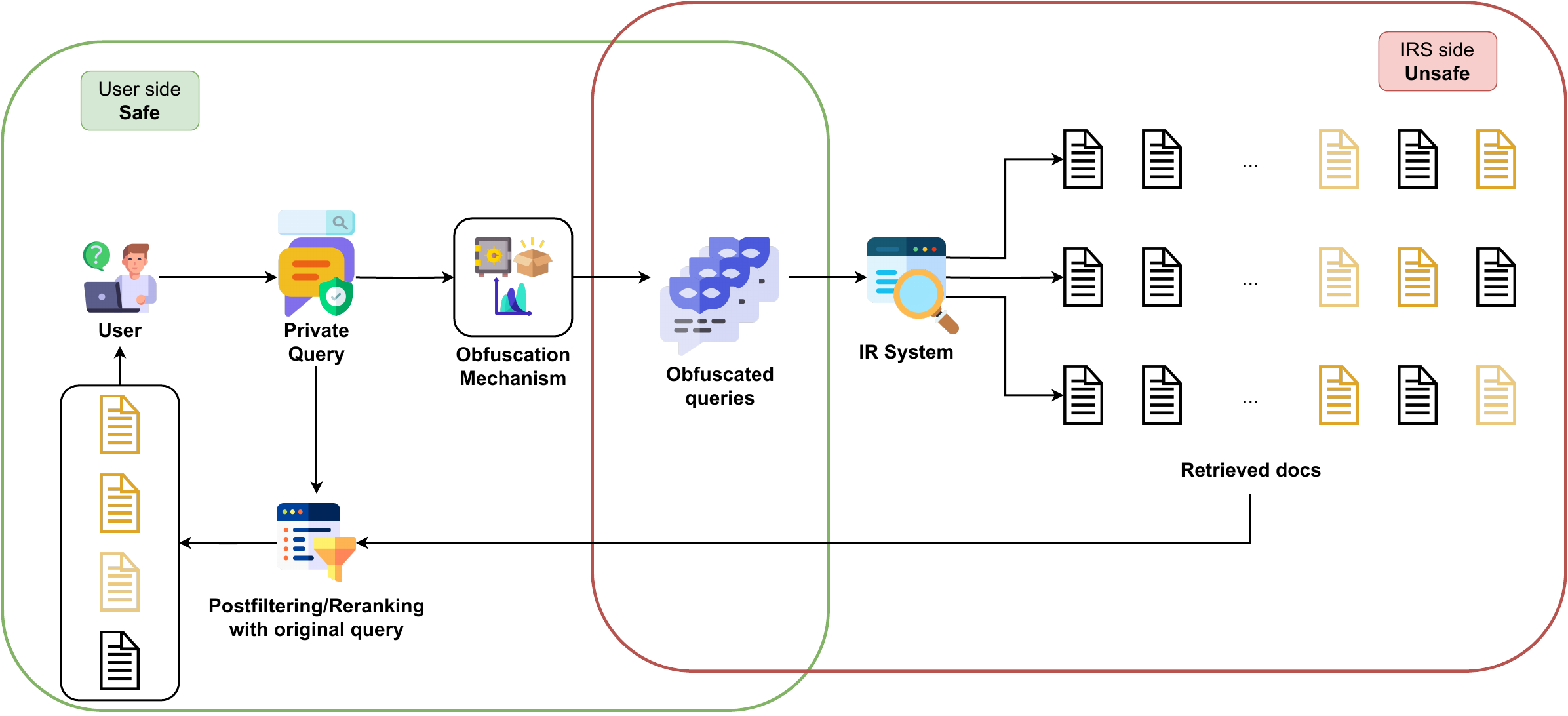}
    \mycaption{General overview of the query obfuscation pipeline in \ac{IR}. The diagram illustrates the pipeline of the retrieval, showing the steps on the User, safe, and \ac{IRS}, unsafe, side.}
    \label{fig:pipelinemoverview}
\end{figure*}

%% file: sections/03methodology.tex
\section{Methodology}
\label{sec:methodology}
In this section, we describe the motivations for this study and the methodology underlying \ac{WBB} to bring concrete privacy for the user's queries, also providing a formal proof of how such a mechanism achieves $\varepsilon$-\ac{DP}.

\subsection{Motivations of the study}
\label{subsec:Language-leakage}
Most of the current approaches designed to obfuscate texts~\citep{xu2020mahalanobis, feyisetan2021private}, including queries in the \ac{IRS} scenario~\citep{FaggioliFerro2024}, as well as approaches to author attribution masking~\citep{weggermann2018syntf}, do consider that, even though the query is obfuscated, it might still reveal too much information by releasing synonyms of the original terms.
In previous work, it has been assumed that using synonyms or hyponyms does not significantly impact safeguarding text privacy~\citep{xu2020mahalanobis,feyisetan2021private,weggermann2018syntf,arampatzis2013queriesscrambler}. However, these words often have a similar meaning to the original terms and they may not offer adequate protection. Suppose a user is issuing the query ``treatment for skin cancer''. Consider now an obfuscation mechanism, as those described in Subsection~\ref{subsec:DP-mechanisms}, that adds a certain noise to the embedding of the term ``cancer", resulting in its mapping to a new obfuscated term, namely ``melanoma". According to the mechanism's perspective, the obfuscation successfully substituted the word ``cancer''. Nevertheless, from a human perspective, the final obfuscated query will be the following: ``treatment for skin melanoma". The adopted approach might be considered as a satisfactory obfuscation, but, in reality, any human would easily understand the actual user's information need, causing severe information leakage. We argue that it is not sufficient to add controlled statistical noise to word embeddings to ensure the privacy of the queries.

Another motivation for a concrete privacy-preserving mechanism is related to cryptography used to protect privacy: cryptographic protocols may not be the appropriate solution for safeguarding privacy in this scenario, as the ciphertext is not directly employed by the \ac{IRS} for document selection and retrieval. Therefore, while cryptographic solutions may protect confidentiality from external eavesdroppers in an effective manner, they do not represent the appropriate solution for privacy in such non-cooperative systems.

\paragraph{Semantically related words in embedding spaces}

Following previous literature~\citep{chen2013expressivnesswe,ono2015antonymdetection}, we analyzed how embeddings of synonyms and hyponyms spread in the vector space and how the embeddings of words belonging to such a vector space of dimension $d$, i.e., the Euclidean space $\mathbb{R}^{d}$, distributes reciprocally. For such a preliminary analysis, we studied a geometric model in $\mathbb{R}^3$ of the relative positioning of two generic vectors: to assess the degree of semantic similarity of two arbitrary words $v$ and $w_1$, where $v$ acts as query term and $w_1$ as obfuscation term, we project them on the non-contextual embedding space $\phi(v)=V$ and $\phi(w_1)=W_1$. We consider two aspects: the Euclidean distance and the angle between the vectors $V$ and $W_1$. To measure these values, we first compute the plane $\pi$ generated by a pair of linearly independent vectors $(e_1,e_2)$, i.e., $\pi=\braket{e_{1},e_{2}}$ containing both the embedding vectors of the words. Figure~\ref{fig:geometricintuition} shows a geometric representation of the situation.

\input{figures/geometry_model}

Figure~\ref{subfig:2dmodel} shows the relative positioning of the 3-dimensional vectors of Figure~\ref{subfig:3dmodel} in bi-dimensional plane $\pi$. We then compute the angle $\alpha$ and the distance between the two vectors to understand how much they differ from each other. We tested multiple words, randomly sampling 10,000 words from the GloVe~\citep{pennington2014glove} dictionary, and, by using different embedding dimensions (50, 100, 200, 300), we observed that, employing the synonyms and hyponyms provided by the Natural Language ToolKit (NLTK)~\citep{bird2004nltk}, most similar words distribute closer considering the Euclidean distance and the angle in the theoretical model. 

\input{figures/radar}
Figure~\ref{fig:radar} presents a radar plot that illustrates the relationship between various synonyms and hyponyms of the word ``Death'', which can be considered sensitive in some scenarios, providing an understanding of the distribution of these words. We noticed that direct hyponyms, represented by a cross ($\times$), and synonyms, represented by a square ($\square$), are in the immediate neighbourhood of the word of interest (the centre of the polar coordinates). Between the centre of the radar and the similar words highlighted, there are many unrelated terms ($\circ$ symbol in Figure~\ref{fig:radar}), i.e. words that are not synonyms and hyponyms of the word of interest. To mitigate the risk of selecting a too similar word as obfuscation term, our proposal consists of creating a ``safe box" around the word embedding, in which no obfuscation words can be sampled. The notion of a safe and candidate box is inspired by literature in the domain of positional masking~\citep{allshouse2010geomasking,hampton2010mapping}. For instance~\citet{hampton2010mapping} proposed to mask geographic points by projecting them on a new location which is randomly selected to fall on a ring, i.e., the space between two concentric disks centred in the original data point, ensuring a minimum of privacy, since the position of the original point is always masked outside the inner disk. Also, we propose to use a second area around the term that needs to be obfuscated, called ``candidate box" so that the terms selected will be similar enough to the original query. Hence, the name of the approach: \acf{WBB}.

\subsection{WBB Mechanism}

Figure~\ref{fig:wbbOverview} illustrates the different steps of the \ac{WBB} mechanism for obfuscating the user query. In the following subsections, we present the details of each obfuscation step. Finally, we provide the mathematical proof of the \ac{DP} property for the proposed mechanism.

\input{figures/wbbOverview}

\subsubsection{Preprocessing}
Firstly, \ac{WBB} preprocesses the original query text. This includes converting the text to lowercase and tokenizing it (i.e., splitting it into isolated terms). Such processing is performed to normalize the text and compute the embeddings accurately.
Subsequently, inspired by  \citet{peng2022semanticspreserved} who observed that perturbing only names and adjectives can significantly enhance the empirical privacy of texts without compromising the performance of the downstream task, we use \ac{POS} tagging to recognize names and adjectives that will be later obfuscated by \ac{WBB}.

Finally, all words are encoded using a non-contextualized encoder $\phi$.

\subsubsection{Mapping Function $f_{mapping}^{(k,n)}$}
\label{subsubsec:mapping_function}
As mentioned above, to identify among which words to sample the obfuscation word, we are taken between two competing needs:
\begin{itemize}
    \item \textbf{Safety}: The non-contextualized representation of the obfuscation term needs to be far enough from the query term representation to make sure it is not a direct synonym/hypernym;
    \item \textbf{Effectiveness}: The non-contextualized representation of the obfuscation term needs to be sufficiently close to maintain a topical relation with the obfuscated term.  
\end{itemize}

To identify which words to consider as a viable candidate for obfuscating a word $w$, we consider three different similarity measures: the cosine similarity, the Euclidean similarity, and the product of the two. We refer to approaches instantiated using each measure respectively as \textit{angle obfuscation}, \textit{distance obfuscation}, and \textit{product obfuscation}. The selected distance functions aim to address the concerns raised in Section~\ref{subsec:Language-leakage}. On the other hand, the product of the two functions is recommended in order to balance their respective impact. Hence, for each word $w'$, we can define $s_{w, w'}=m(\phi(w), \phi(w'))$ the similarity between the original word $w$ using one of the aforementioned similarity measures. Then, the mechanism sorts the words based on the similarity results computed.

Once words have been obfuscated, we can define a function  $f^{(k,n)}_{\text{mapping}}$ that produces the candidate set $\mathcal{C}$ such that the candidate contains the $n$ closest words to the original word $w$, whose similarity is above the similarity of the $k$-th most similar term. By excluding the first top-$k$ words, we ensure that all the words are sufficiently different from the original ones. At the same time, the dimension of the candidate set, i.e., $n$, allows one to choose a set that contains many different words. One of the major advantages of this procedure is that the parameters can be tuned according to the user's privacy concern: a very concerned user can adopt a large $k$ and $n$: the query will be very different from the original one. Therefore, we can expect high privacy but, at the same time, possibly low retrieval performance. On the contrary, a not-concerned user can use a small $k$ (even 0, if they accept that a word can be obfuscated with itself) and $n$, achieving higher effectiveness but lower privacy. Furthermore, the presence of these bounding boxes to sample obfuscated words ensures that the $\varepsilon$ of the \ac{WBB} mechanism (described in the next section) cannot cause catastrophic data releases, in which the query is obfuscated with the query itself, achieving privacy only from a formal but not a practical point of view.


\subsubsection{Sampling Function $f_{sampling}^{(k,n,\varepsilon)}$}
The last phase consists of sampling the substitution words among the candidates present in $\mathcal{C}$ for each word to obfuscate, using the function $f^{(k,n,\varepsilon)}_{\text{sampling}}$. 

Such sampling relies on the exponential mechanism described in Section~\ref{sec:theoreticalBackground}. To instantiate the exponential mechanism, we need a utility function $u$. More in detail, given $s_{w,w'}$ the similarity score between words $w$ and $w'$, we compute for each $w'\in\mathcal{C}$ the $Z_{score}$ as $Z_{\text{score}}=\frac{s_{w,w'}-\mu}{\sigma}$, where $\mu$ is the average similarity over the words $w'\in\mathcal{C}$ and $\sigma$ the standard deviation of the similarities in the set $\mathcal{C}$ provided after the mapping step of the mechansim. 

Finally, we normalize such scores into the interval $\left[0,1\right]$, computing the final utilities using the function in Equation~\ref{eq:utility}.
The utility function is defined as follows:
\begin{equation}
    u(w,w') = \frac{1}{1+\exp\left(Z_{\text{score}}\right)}
    \label{eq:utility}
\end{equation}

We want to underline that our way of employing \ac{DP} is fundamentally different from the one currently adopted by other approaches addressing similar tasks, as described in Subsection \ref{subsec:DP-mechanisms}. Indeed, those approaches employ \ac{DP} to perturb the embedding vectors, but by setting the value of $\varepsilon$ to be sufficiently large, the noise becomes so small that the produced query is exactly the same as the original one. In our case, we use \ac{DP} to sample from a set of words that, by construction, are safer as the set does not contain the original word nor highly similar terms, according to their non-contextualized embeddings.

\subsubsection{Formal proof of \ac{DP} for \ac{WBB}}
\label{subsec:formalproof}
\input{eqs_algs/mechanism}

To assemble all the steps together, Algorithm~\ref{alg:mechanism} provides the \ac{WBB} pseudo-code. The proof of the $\varepsilon$-DP stems from its exponential mechanism.

\begin{theorem}
\label{theo:dp}
    The mechanism $\mathcal{M}$ explained in Algorithm \ref{alg:mechanism} is $\varepsilon$-Differentially Private.
\end{theorem}
\begin{proof}
    Let $\mathcal{X}$ be an input set that is mapped into the corresponding output set $\mathcal{Y}$ by the mechanism $\mathcal{M}$. For any pair input $x, x'\in \mathcal{X}$, identifying with $Z_{\text{score}}$ and $Z'_{\text{score}}$ the respective standardized initial scores, it holds that using the definition of utility function introduced in Equation~\ref{eq:utility}, the sensitivity $\Delta u$ is bounded by 1, since:
    \begin{equation}
        \begin{split}
        \Delta u &= \max_{y\in\mathcal{Y}}\max_{x,x'\in\mathcal{X}}|u(x,y)-u(x', y)|\\
&=\max_{y\in\mathcal{Y}}\max_{x,x'\in\mathcal{X}}\left|\frac{1}{1+\exp\left(Z_{\text{score}}\right)}-\frac{1}{1+\exp\left(Z'_{\text{score}}\right)}\right|\\
        &\leq1
        \end{split}
\end{equation}
    Hence, we can follow a reasoning similar to the proof of the exponential mechanism itself for concluding the proof~\citep{mcsherry2007exponential}. The sampling function of the Algorithm~\ref{alg:mechanism} provides the privacy property of the mechanism. Given a privacy budget $\varepsilon>0$, the probability of sampling an output $y$ given an input $x$ is given by the exponential mechanism, hence:
    \begin{equation}
        \text{P}\left[f_{\text{sampling}}(x)=y\right] = \frac{\exp\left(\frac{\varepsilon u(x,y)}{2\Delta u}\right)}{\sum_{\tilde{y}\in\mathcal{Y}}\exp\left(\frac{\varepsilon u(x,\tilde{y})}{2\Delta u}\right)}
    \end{equation}
    Finally, we conclude the proof with the following inequality, using the exponential mechanism employed in the sampling function:
    \begin{equation}
        \begin{split}
            & \frac{\text{P}\left[f_{\text{sampling}}(x)=y\right] }{\text{P}\left[f_{\text{sampling}}(x')=y\right] } = \frac{\frac{\exp\left(\frac{\varepsilon u(x,y)}{2\Delta u}\right)}{\sum_{\tilde{y}\in\mathcal{Y}}\exp\left(\frac{\varepsilon u(x,\tilde{y})}{2\Delta u}\right)}}{\frac{\exp\left(\frac{\varepsilon u(x',y)}{2\Delta u}\right)}{\sum_{\tilde{y}\in\mathcal{Y}}\exp\left(\frac{\varepsilon u(x',\tilde{y})}{2\Delta u}\right)}} \\
            &= \exp\left(\frac{\varepsilon \left(u(x,y)-u(x',y)\right)}{2\Delta u}\right)\frac{\sum_{\tilde{y}\in\mathcal{Y}}\exp\left(\frac{\varepsilon u(x',\tilde{y})}{2\Delta u}\right)}{{\sum_{\tilde{y}\in\mathcal{Y}}\exp\left(\frac{\varepsilon u(x,\tilde{y})}{2\Delta u}\right)}} \\
            &\leq \exp\left(\frac{\varepsilon \left(\Delta u\right)}{2\Delta u}\right)\frac{\sum_{\tilde{y}\in\mathcal{Y}}\exp\left(\frac{\varepsilon u(x',\tilde{y})}{2\Delta u}\right)}{{\sum_{\tilde{y}\in\mathcal{Y}}\exp\left(\frac{\varepsilon u(x,\tilde{y})}{2\Delta u}\right)}} \\
            &\leq \exp\left(\frac{\varepsilon}{2}\right) \exp\left(\frac{\varepsilon}{2}\right) \frac{\sum_{\tilde{y}\in\mathcal{Y}}\exp\left(\frac{\varepsilon u(x,\tilde{y})}{2\Delta u}\right)}{{\sum_{\tilde{y}\in\mathcal{Y}}\exp\left(\frac{\varepsilon u(x,\tilde{y})}{2\Delta u}\right)}} = e^\varepsilon
        \end{split}
    \end{equation}
\end{proof}
This terminates the formal bases of the proposed WBB method and proves the $\varepsilon$-Differential Privacy of the mechanism in Algorithm~\ref{alg:mechanism}.

%% file: figures/geometry_model.tex
\begin{figure}[h!]
        \centering
            \subfigure[Model 3D - Vector positioning of $V$ and $W_{1}$.] 
            {
                \label{subfig:3dmodel}
                \includegraphics[width = 0.45\textwidth]{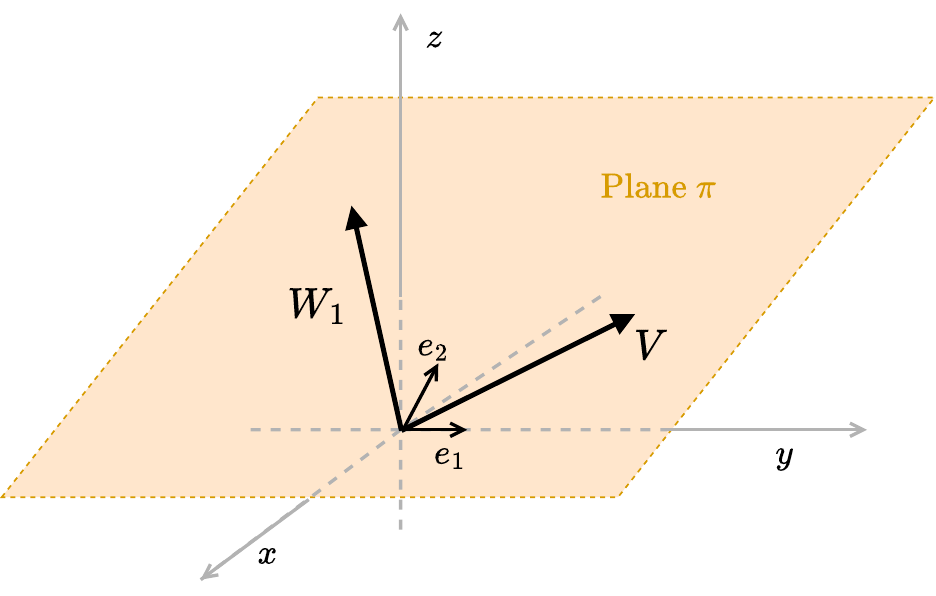} 
            }
            \subfigure[Model 2D - Relative positioning of $V$ and $W_{1}$ in the plane $\pi$.] 
            {
                \label{subfig:2dmodel}
                \includegraphics[width = 0.45\textwidth]{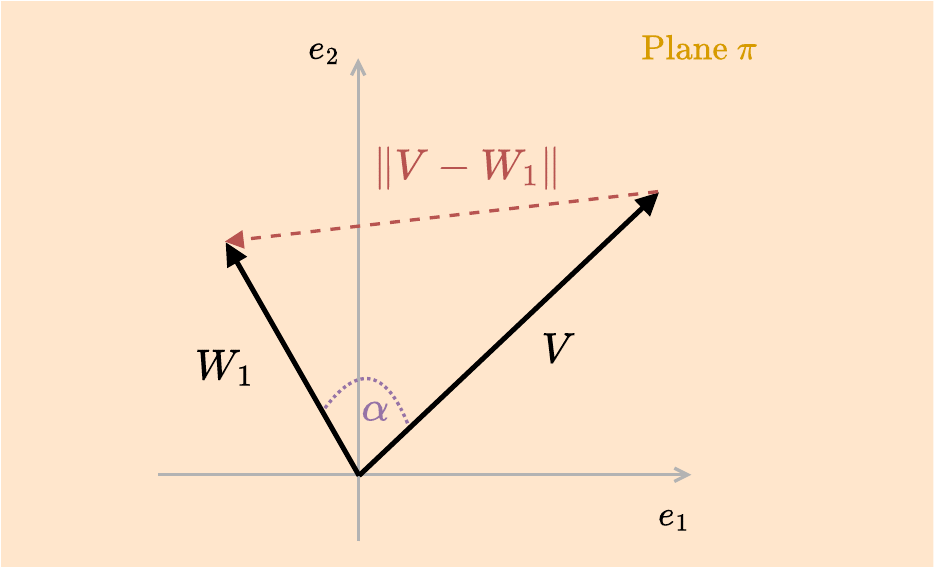} 
            }
         \mycaption{Geometric intuition of the vector space.}
        \label{fig:geometricintuition}
\end{figure}

%% file: figures/radar.tex
\begin{figure}[h!]
    \centering
    \includegraphics[width=1\linewidth]{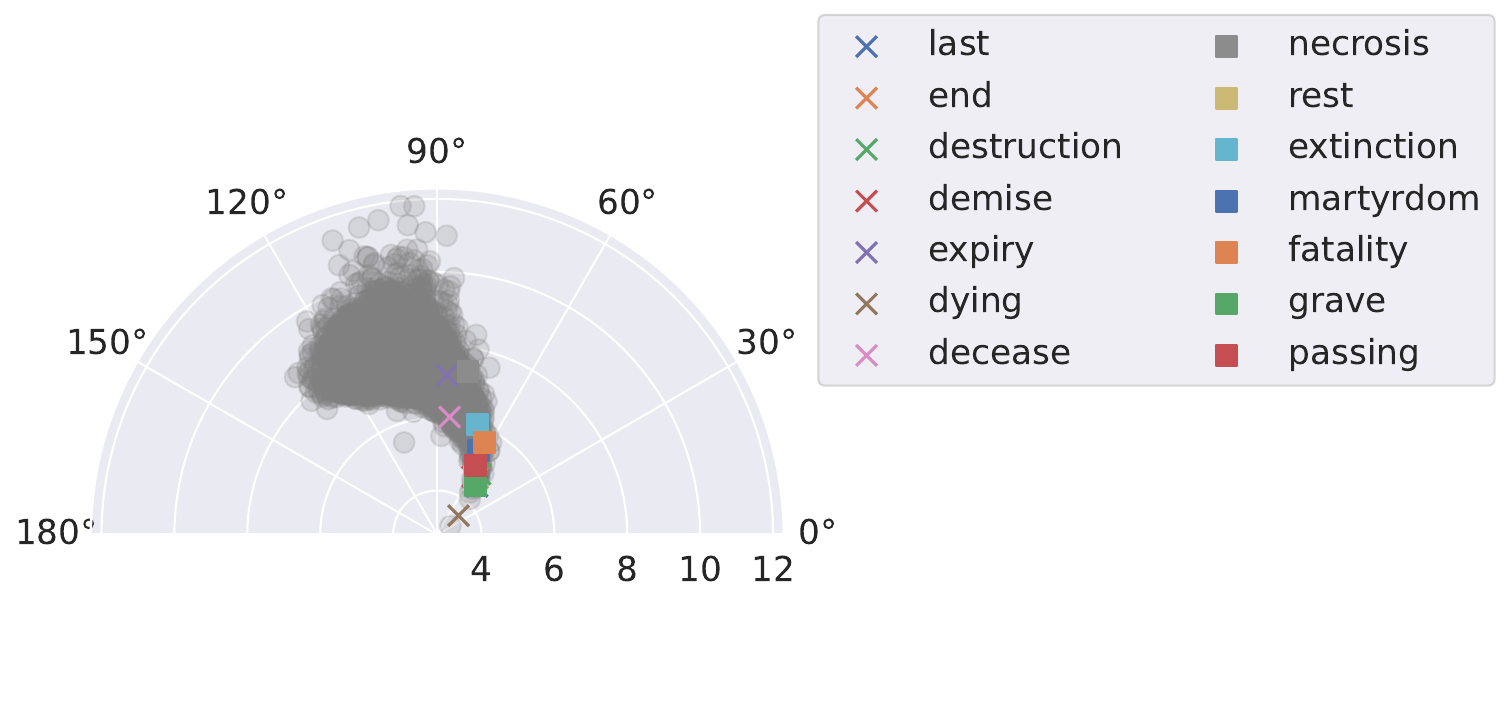}
    \mycaption{Radar plot showing the distribution of sampled words and the word ``Death" considering Euclidean distance and angle. The symbols of crosses ($\times$) and squares ($\square$) are utilized to represent the linguistic relations of a given word. Specifically, the crosses indicate the hyponyms,  while the squares represent the synonyms. The grey circles ($\circ$) represent other words that are neither hyponyms nor synonyms.}
    \label{fig:radar}
\end{figure}

%% file: figures/wbbOverview.tex
\begin{figure}[h!]
    \centering
    \includegraphics[width=1\linewidth]{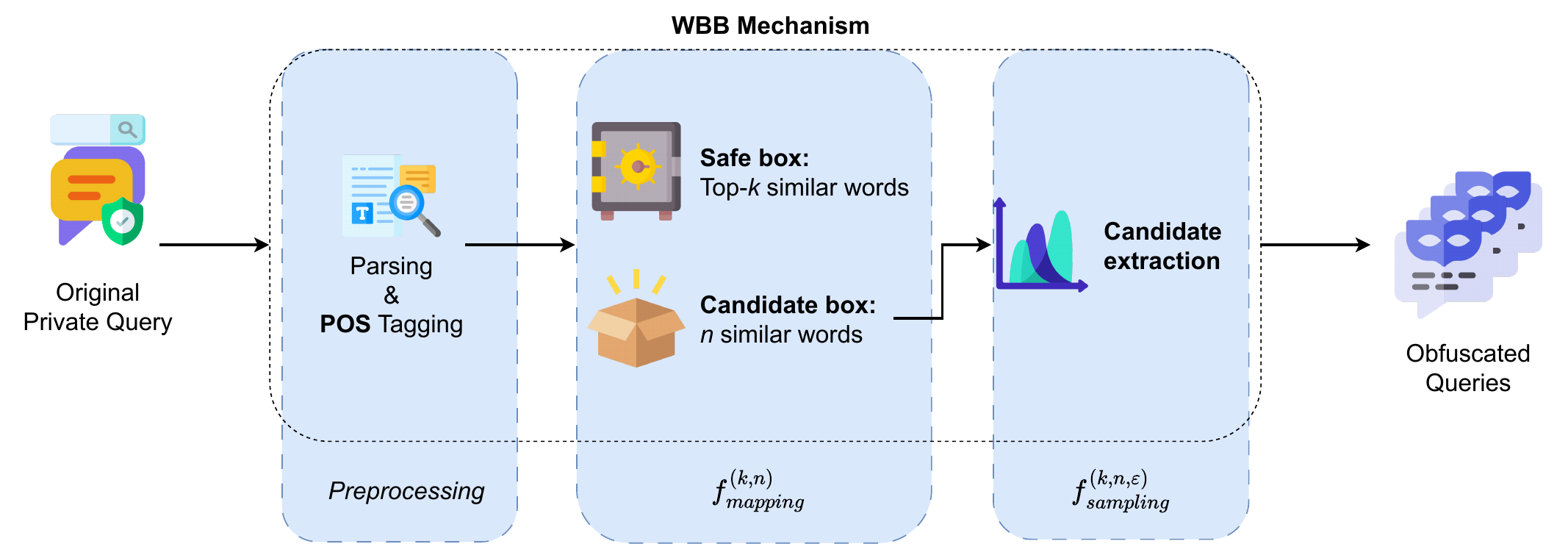}
    \caption{\ac{WBB} mechanism schematic overview of the obfuscation procedure. Compared to the pipeline presented in Figure~\ref{fig:pipelinemoverview}, this is the ``Obfuscation Mechanism'' component.}
    \label{fig:wbbOverview}
\end{figure}

%% file: eqs_algs/mechanism.tex
\begin{algorithm}[h!]
\DontPrintSemicolon
  \KwInput{Original query $Q=\langle w_{0}, \dots, w_{\ell}\rangle$, $\varepsilon>0$, $n>0$, $k>0$}
  \KwOutput{Obfuscated query $\tilde{Q}=\langle \tilde{w}_{0}, \dots, \tilde{w}_{\ell}\rangle$}
  \KwData{Vocabulary $V$, set of \ac{POS} tags $S$, function $m$}
  Define $\tilde{Q}=\langle\rangle$\;
  {{\footnotesize\ttfamily\textcolor{blue}{// Preprocessing}}}\;
  Tokenize the query $Q$, and \ac{POS} tag each word $w_{i}$\;
  \For{$w_{i}\in Q$} 
    {
    \If{\texttt{tag}$(w_{i})\in S$\;}
        {
            {{\footnotesize\ttfamily\textcolor{blue}{// Mapping Function: $f_{mapping}^{(k,n)}$}}}\;
            get $s_{i}=m(w_{i}, v)\ \forall v\in V$\;
            sort $s_{i}$ and get their position (\text{\texttt{pos}})\;
            \texttt{safe box} $=\{v\in V: \text{\texttt{pos}}(v)\leq k\}$\;
            $\mathcal{C}=\{v\in V: k <\text{\texttt{pos}}(v)\leq n+k\}$\;
            {{\footnotesize\ttfamily\textcolor{blue}{// Sampling Function: $f_{sampling}^{(k,n,\varepsilon)}$}}}\;
            get $Z_{\text{score}}\ \forall w \in \mathcal{C}$\;
            get $u_i\ \forall w \in \mathcal{C}$ according to Equation~\ref{eq:utility}\;
            $w_i = \tilde{w}_{i}$ where $f_{sampling}^{(k,n,\varepsilon)}(u_i)=\tilde{w}_{i}$ from $\mathcal{C}$\;
        }
    Append $w_{i}$ to $\tilde{Q}$\;
    }
    \Return $\tilde{Q}$\;
\caption{WBB Mechanism.}
\label{alg:mechanism}
\end{algorithm}

%% file: sections/04experimentalAnalysis.tex
\section{Experimental Analysis}
\label{sec:experimentalAnalysis}
This section provides an in-depth analysis and discussion of the results obtained through the experiments.

\subsection{Experimental Setup}
As encoder function $\phi$, we use the embeddings of GloVe~\citep{pennington2014glove} trained on the Common Crowd as vocabulary for our investigations. We experimented with vectors of size 300; other sizes have comparable behaviour. We test the obfuscation approach considering two different collections: TREC Deep Learning (DL'19)~\citep{msmarcodl19}, based on the MS MARCO~\citep{msmarco} passages corpus which contains 43 queries, and TREC Robust ‘04~\citep{robust04collection}, which relies on disks 4 and 5 of the TIPSTER corpus, minus congressional records, and contains 249 queries. Additionally, we consider four retrieval models, considering two sparse bag-of-word models, namely BM25~\citep{bm25} and Vector Space Model (TF-IDF)~\citep{tfidf}, and also two dense bi-encoders, TAS-B~\citep{hofstatter21tasb} and Contriever~\citep{izacard22contriever}. These models represent our non-cooperating \ac{IRS}. As mentioned in Section \ref{sec:methodology}, the results are then merged and re-ranked using Contriever, simulating a safe re-ranking on the user's side. For each obfuscation mechanism and collection, we generate 20 obfuscated queries to retrieve the top-$100$ documents. 

For reproducibility purposes, we release the code we used at \url{https://github.com/Kekkodf/WBB-QueryObfuscation}.

\subsubsection{Baseline Approaches}
To evaluate our mechanism, we compare performances and privacy with two different obfuscation mechanisms, namely the \ac{CMP}~\citep{feyisetan2020multivariate} and the \ac{MHL}~\citep{xu2020mahalanobis} mechanisms. 
As described in Section~\ref{subsec:DP-mechanisms}, \ac{CMP} and \ac{MHL} obfuscate each word independently and merge them in a single query. The obfuscation is achieved by adding some noise to the non-contextual embedding of each word and taking the closest word to the noisy embedding. In particular, \ac{CMP} samples the noise from a multivariate Laplace distribution. \ac{MHL}, on the other hand, uses the Mahalanobis norm to modify the direction in which the noise is sampled, such that it is more likely to sample denser areas of the embedding space to increase the chances that a word is not obfuscated by itself. We also compare our results with other state-of-the-art mechanisms which do not guarantee a formal $\varepsilon$-\ac{DP}, but that have been explicitly designed for \ac{IR}, i.e., the methods proposed by~\citet{arampatzis2011enenhcingdeniabilityquerylogs} (referred to as \aea) and~\citet{frobe2022efficentqueryobf} (\fea).

\subsection{Results}

\subsubsection{Privacy Guarantees} 
\label{subsec:privacyResults}
As detailed in Subsection~\ref{subsec:Privacy-metrics}, the metrics commonly employed to assess privacy levels are the \textit{Uncertainty Statistics}, \textit{Lexical Similarity}, and \textit{Semantic Similarity}. However, classical measures for the \ac{WBB} mechanism are not quite effective using the \textit{Uncertainty Statistics}: by mechanism definition, such uncertainty measures, i.e., the mechanism failure rate $N_{w} = \text{Pr}\left[\mathcal{M}(w)=w\right]$ and $S_{w} = \min\left|\left\{\mathcal{S}\subseteq \text{Image}(\mathcal{M}): \text{Pr}\left[\mathcal{M}(w)\notin \mathcal{S}\right]\leq\eta\right\}\right|$ for the number of words to which the mechanism obfuscate a term, will always be 0 as the \ac{WBB} mechanism preserves the confidentiality of a term and never obfuscates it with itself. Notice that implicitly $N_w$ describes the proportion of failures: ``How many times did the obfuscation mechanism not obfuscate, leaving the user exposed?", we argue that even a single time is too much. Similarly, $S_w$ is proportional to the \ac{WBB} parameters $k$ and $n$, and thus it is not informative. In the \ac{WBB} mechanism, the measure $S_w$ resembles more the $k$-anonymity property~\citep{Sweeney2002kanonimity}. Given a Candidate box $\mathcal{C}$ of size $n$, the probability of re-identification of a term without prior knowledge is equal to $\frac{1}{n}$, therefore enlarging the number of possible candidates $n$ the probability of guessing the correct term reduces to 0 even for small privacy budgets $\varepsilon$.

Thus, our privacy analysis focuses on the Lexical and Semantic Similarity measures to evaluate the quality of the obfuscated queries produced. To evaluate the lexical similarity, we used the Jaccard Similarity. For the semantic similarity, we employ as embedding function MiniLM~\citep{reimers2020multilingualsentencebert} using the L6 v2 version available on HuggingFace\footnote{\url{https://huggingface.co/sentence-transformers/all-MiniLM-L6-v2}} which maps the sentence to a 384-dimensional dense vector space and computes the cosine similarity between the original and the obfuscated encodings of the queries. To evaluate the effectiveness of the \ac{WBB} mechanism, we compared it against existing state-of-the-art mechanisms~\citep{feyisetan2020multivariate,xu2020mahalanobis}. We modified the different parameters of the \ac{WBB} mechanism, namely the distance function and the parameters $(k,n)$. The results are reported in Table~\ref{tab:CosineSimilarity}, where the lower the Jaccard and semantic similarities are, the higher the privacy provided.
\input{tables/cosineSimilarity}
\paragraph{Discussion}
We observe that the Jaccard similarity for State-of-the-Art mechanisms is proportional to the privacy budget $\varepsilon$, suggesting that these mechanisms -- especially for large $\varepsilon$, struggle to replace the original words due to the insufficient noise added. Nonetheless, \ac{WBB} guarantees that words identified as sensitive and captured by the \ac{POS} tagger during preprocessing are always masked -- such words fall into the safe box and cannot be used for masking. Consequently, the intersection between original and obfuscated words invariably remains of size 0, ensuring that the Jaccard similarity is consistently 0.

Regarding the Semantic Similarity, we need a small step back to address one of the major limits of \ac{DP}: a notable critique levelled against $\varepsilon$-\ac{DP} is that as the privacy budget $\varepsilon$ increases, the privacy assurances offered by the mechanism vanish \citep{dwork2014privacybook,domingoferrer2021limitsdp}. Indeed, when the value of $\varepsilon$ increases, the State-of-the-Art mechanisms, i.e., \ac{CMP} and \ac{MHL}, struggle in the obfuscation process, suggesting that at high $\varepsilon$ values, only formal privacy is maintained. An important aspect handled by the \ac{WBB} mechanism through computing the Semantic Similarity values is that the privacy budget $\varepsilon$ does not impact the sentence similarity of the obfuscated queries. This characteristic overcomes the limitation of formal privacy because \ac{WBB} not only ensures the formal guarantee of the $\varepsilon$-\ac{DP} property, as proved in Section~\ref{subsec:formalproof}, but also provides an additional layer of concrete privacy to the original query. Specifically, for all instances analysed, the average cosine similarity computed is always below 0.5, meaning that the \ac{WBB} mechanism provides a true obfuscation of the original queries. Only in the case of a small value of $k$, i.e., the candidate box size, the cosine similarity approaches the 0.5 value.

Moreover, the most severe privacy strategies lowering the semantic similarity are the ones that use as distance function the Euclidean distance and the one that proposes a large size candidate box, i.e., \ac{WBB}$(4,50)$ \textbf{D} and \ac{WBB} $(4,250)$ \textbf{A}. This fact suggests that as we enlarge the candidate box, or on the other hand, we enforce the Euclidean Distance as a distance function to compute the mechanism boxes, more privacy is provided at the same level of formal privacy.

\subsubsection{Retrieve Relevant Documents}
\input{tables/recall}

As multiple obfuscation queries are sent to the \ac{IRS}, our first measure of interest is the pooled recall over all obfuscation queries, i.e., the recall considering the set of unique documents retrieved by any obfuscation query. This allows us to verify that the \ac{WBB} approach can bring enough relevant documents to the user's side. Table~\ref{tab:recall_results} reports the observed merged recall. Based on the results of the privacy evaluation, Subsection~\ref{subsec:privacyResults}, we report the \ac{WBB} retrieval capabilities based on the angle obfuscation.

There are two different patterns, depending on whether we assume the black box \ac{IRS} to be based on a lexical (BM25 and TF-IDF) or semantic \ac{IRS} (Contriever and Tas-B). In the first case, the recall of WBB is lower than those of other \ac{DP} mechanisms and \aea. At the same time, the recall of WBB is in line with the one of \fea\ on Robust `04, and lower on DL `19. 
If we switch to semantic \ac{IRS}, the recall remains overall lower compared to other \ac{DP} mechanisms and \aea, but the margin is far smaller. Nevertheless, when compared with \fea\ for semantic \ac{IRS}, the performance of WBB is higher.
The reason behind this behaviour is explained by the fact that lexical \ac{IRS} rely on exact matching. By completely preventing the usage of query terms on the obfuscated queries, both WBB and \fea\ suffer in terms of effectiveness if the \ac{IRS} is lexical. 
Vice versa, the phenomenon is less evident regarding semantic systems, as the WBB relies on the candidate box: terms that are sufficiently close in the latent embedding space to the original ones. Thus, we can expect some form of topical relation between real and obfuscation terms that allows semantic \ac{IRS} to be less affected by the privacy protection measures put in place by WBB.
When it comes to the behaviour of the different \ac{DP} approaches concerning the $\varepsilon$, we notice a completely different pattern between classical approaches based on introducing noise on the word embeddings (i.e., CMP and Mhl) and our solution based on sampling the words according to their similarity in the embedding space.
Both CMP and Mhl are severely affected by changes in $\varepsilon$.
In particular, for $\varepsilon=1$ and $\varepsilon=5$, depending on the scenario, the recall is on par with or lower than the one by WBB. Then, starting with $\varepsilon=10$ up to $\varepsilon=20$, we observe a recall even higher than the one achieved by the original \ac{IRS} with no privacy measures in place. The same pattern was also observed by~\citet{FaggioliFerro2024}. The behaviour is explained by the fact that with such big $\varepsilon$, several terms in the obfuscated query will be the same as the ones in the original query. In contrast, the changed ones are often synonyms and thus cause some sort of query expansion. Finally, if $\varepsilon=50$, the recall is almost identical to the one without privacy: with such a $\varepsilon$, \ac{DP} does not obfuscate at all, becoming simply a formal definition. 
Overall, WBB reacts to changes to $\varepsilon$, causing changes to recall in a quite bounded range. While this might be a limitation for more advanced users who wish to have the capacity to alter the system's behaviour completely, it might be an advantage for average users who are not at risk of ending up in catastrophic conditions. Finally, we changed the values of the parameters $(k,n)$ of the \ac{WBB} mechanism, and as suggested by \citet{FaggioliFerro2024}, we analyzed the average recall measured using as privacy budget $\varepsilon=10$. Figure~\ref{fig:confmatrix_Recall} reports the results obtained. 

\input{figures/cm_combined_R}
\paragraph{Discussion} Considering the different distance functions: as for the privacy analysis in Subsection~\ref{subsec:privacyResults}, the \ac{WBB} configurations that work the better are the angle and product-based obfuscation: the initial finding that we obtained when retrieving documents from the \ac{IRS} is that, stemming naturally from the privacy analysis conducted in Subsection~\ref{subsec:privacyResults}, the \ac{WBB} mechanism configuration most significantly affected in terms of recall is the one using the Euclidean distance function for computing safe and candidate box. In addition, another important trend that can be derived from Figure~\ref{fig:confmatrix_Recall} is that the lower the candidate box, the lower the recall: especially for the \ac{WBB}(8,5) all the computed recalls are below 0.25. This suggests that from a narrow candidate box and eliminating a high number of similar words, there are high privacy guarantees but low retrieval performance. 

On the other hand, a further consistent pattern across all the comparison matrices is that the performance tends to increase as a higher number of candidates is allowed, particularly when the safe box is not excessively wide. The higher results are obtained when \ac{WBB} is configured with only two words in the safe box, achieving a recall equal to 0.384 for \ac{WBB}(2,20) and 0.375 for \ac{WBB}(2,20), with a distance function angle and product based respectively.

\subsubsection{Determining the Utility Preserved}

\input{tables/nDCG_NEW}
As aforementioned, once the documents are available on the user's side, they need to be re-ranked using the real information needs. To implement this, we use Contriever as the reranker.
Table~\ref{tab:nDCG_results} reports the nDCG@10 values after re-ranking the documents retrieved through obfuscation queries by each \ac{IRS}. 
We report the results for $\varepsilon=10$, as WBB is not particularly affected by it, and it is one of the recommended values according to~\citet{FaggioliFerro2024}.
In line with what was observed for Table~\ref{tab:recall_results}, WBB appears to struggle if the underneath \ac{IRS} relies on exact matching, for the reasons previously mentioned.
Nevertheless, it is also important to mention that the obfuscation queries generated by the best-performing methods for Robust `04, \ac{CMP} and \ac{MHL} have a Jaccard similarity with the original query, respectively of 0.225 and 0.101: one out 5 and one out of 10 words, respectively are left unchanged. As the original queries have, on average, 5.8 words for CMP, in expectation, slightly less than one term is left exactly as it is in every query. This means that, by sending 20 queries to the \ac{IRS}, we can expect that each real term will appear 4 times across the set of obfuscated queries. The same reasoning applies straightforwardly to \aea, whose generated obfuscation queries have a Jaccard similarity of 0.200. 
We argue that, by showing the real terms of the query to the \ac{IRS} multiple times scattered across obfuscation queries, it is easy to recognize them.
Vice-versa, both WBB and \fea suffer in terms of performance when the \ac{IRS} uses lexical matching, but they are guaranteed to not show the real query terms, nor excessively related terms.
If we compare WBB and \fea, \fea outperforms WBB on Robust `04 with BM25 and TF-IDF by 0.030 and 0.031 nDCG points respectively. Vice versa, on DL `19, WBB outperforms \fea by 0.170 nDCG points on both BM25 and TF-IDF.  
The pattern completely changes when we move to the semantic \ac{IRS}. While the privacy guarantees remain the same, i.e., the obfuscation queries are always the same, the obfuscation terms used by WBB are semantically related in a latent space, thanks to the usage of the candidate box. 
As a consequence, when semantic \ac{IRS} are in place, the more secure WBB obfuscation queries are also able to provide high performance, surpassed only by CMP on Robust `04 by a small margin (0.016 nDCG points), while being the most effective solution on DL `19.

\paragraph{Discussion}
\input{figures/cm_combined_nDCG_at_10}

As mentioned in Section~\ref{sec:methodology}, WBB can be instantiated using multiple similarity functions to operationalize the exponential mechanism. Furthermore, it relies on two parameters: $k$, the size of the size box, and $n$, the size of the candidate box. Therefore, we propose here an ablation study of such parameters.
Figure~\ref{fig:confmatrix} reports the nDCG values when varying the parameters if Contriever is used both as \ac{IRS} and re-ranker for DL `19. Similar results can be observed in other scenarios and thus are not reported. A lighter green indicates better performance, while a darker blue corresponds to worse performance. Interestingly, the obfuscation using the Euclidean distance obtains the worst results. Most likely, this form of obfuscation tends to be particularly severe: while we have a decrease in performance, we also have an increase in privacy as supported by the results in Table~\ref{tab:CosineSimilarity} using \textbf{D} as a distance function. Therefore, this obfuscation is suggested for particularly privacy-concerned users.

On the contrary, angle obfuscation, based on the cosine similarity between the real and the candidate obfuscation words, is the most effective solution and is suitable for less privacy-concerned individuals. Finally, the product between the two has an intermediate behaviour suited for intermediate users.
For what concerns $k$, the size of the safebox, as a general trend, the bigger it is, the lower the performance. This behaviour is reasonable, considering that we increase privacy by excluding more and more terms similar to the query terms. Finally, the size of the candidate box has a positive impact on the performance. If the candidate box is too small, then the \ac{DP} mechanism will always sample across the same set of terms. If none of such terms is sufficiently topically related, then the final performance will be low. Vice-versa, a larger candidate box ensures the variability of the sampling, further increasing privacy, but also the likelihood that a sufficiently topically related term will be available.

%% file: tables/cosineSimilarity.tex
\begin{table*}[tb]
\mycaption{Average Jaccard and Semantic Similarities of the obfuscated queries for the collection MSMARCO DL'19~\citep{msmarcodl19}. The analysis comprises the State-of-the-Art results, the WBB$(k,n)$ using different distance functions, i.e., \textbf{A} and \textbf{D} for the cosine and the Euclidean Distance, and \textbf{P} for the product between the two to create safe and candidate box, and the WBB with fixed sizes for safe $(\boldsymbol{k})$ and candidate $(\boldsymbol{n})$ box.}
\label{tab:CosineSimilarity}
\centering
\resizebox{\textwidth}{!}{
\begin{tabular}{@{}llccccccccccccccccccc@{}}
\toprule
 & \multicolumn{1}{c}{} & \multicolumn{8}{c}{\textit{\textbf{Jaccard Similarity}}} & & & \multicolumn{8}{c}{\textit{\textbf{Semantic Similarity}}} \\ 
 
 \midrule
 
 & \multicolumn{1}{c}{} & \multicolumn{8}{c}{$\varepsilon$} & & &\multicolumn{8}{c}{$\varepsilon$} \\ \midrule
 & \multicolumn{1}{c|}{\textbf{Mechansim}} & 1 & 5 & 10 & 12.5 & 15 & 17.5 & 20 & \multicolumn{1}{c}{50}& No DP & 1 & 5 & 10 & 12.5 & 15 & 17.5 & 20 & 50 & No DP\\

\midrule

\multicolumn{1}{l|}{\multirow{4}{*}{\textit{SotA}}}  & \multicolumn{1}{l|}{\aea} & - & - & - & - & - & - & - & - & 0.338 & - & - & - & - & - & - & - & - &  0.509\\
\multicolumn{1}{l|}{} & \multicolumn{1}{l|}{\fea} & - & - & - & - & - & - & - & - & 0.  & - & - & - & - & - & - & - & - & 0.077\\
\multicolumn{1}{l|}{} & \multicolumn{1}{l|}{CMP} & 0. & 0.003 & 0.100 & 0.277 & 0.509 & 0.707 & 0.814 & \multicolumn{1}{c}{0.935} & -  & 0.017 &   0.029 &   0.195 &   0.417 &   0.645 &  0.803 &  0.858 &  0.907 & - \\
\multicolumn{1}{l|}{} & \multicolumn{1}{l|}{Mhl} & 0. & 0.002 & 0.054 & 0.140 & 0.288 & 0.453 & 0.616 & \multicolumn{1}{c}{0.935} & - & 0.020 &   0.028 &   0.114 &   0.240 &   0.420 &  0.592 &  0.730 &  0.910 & -\\
 
 \midrule
 
\multicolumn{1}{l|}{\multirow{3}{*}{\textit{Dist. Meas.}}} & \multicolumn{1}{l|}{WBB$(4,50)$ \textbf{A}} & 0. & 0. & 0. & 0. & 0. & 0. & 0. & \multicolumn{1}{c}{0.} & -  & 0.425 & 0.425 & 0.430 & 0.430 & 0.432 & 0.427 & 0.436 & 0.428& - \\
\multicolumn{1}{l|}{} & \multicolumn{1}{l|}{WBB$(4,50)$ \textbf{D}} & 0. & 0. & 0. & 0. & 0. & 0. & 0. & \multicolumn{1}{c}{0.}& - & 0.366 & 0.370 & 0.367 & 0.366 & 0.361 & 0.368 & 0.368 & 0.366 & -\\
\multicolumn{1}{l|}{} & \multicolumn{1}{l|}{WBB$(4,50)$ \textbf{P}} & 0. & 0. & 0. & 0. & 0. & 0. & 0. & \multicolumn{1}{c}{0.}& - & 0.419 & 0.423 & 0.425 & 0.427 & 0.421 & 0.425 & 0.424 & 0.426 & -\\

\midrule

\multicolumn{1}{l|}{\multirow{5}{*}{\textit{$k$}}} & \multicolumn{1}{l|}{WBB$(\boldsymbol{2},50)$ A} & 0. & 0. & 0. & 0. & 0. & 0. & 0. & \multicolumn{1}{c}{0.}& -  & 0.439 & 0.441 & 0.436 & 0.442 & 0.439 & 0.431 & 0.436 & 0.441& - \\
\multicolumn{1}{l|}{} & \multicolumn{1}{l|}{WBB$(\boldsymbol{6},50)$ A} & 0. & 0. & 0. & 0. & 0. & 0. & 0. & \multicolumn{1}{c}{0.}& -& 0.425 & 0.425 & 0.430 & 0.430 & 0.432 & 0.427 & 0.436 & 0.428& - \\
\multicolumn{1}{l|}{} & \multicolumn{1}{l|}{WBB$(\boldsymbol{12},50)$ A} & 0. & 0. & 0. & 0. & 0. & 0. & 0. & \multicolumn{1}{c}{0.}& - & 0.413 & 0.413 & 0.410 & 0.410 & 0.412 & 0.410 & 0.414 & 0.416& - \\
\multicolumn{1}{l|}{} & \multicolumn{1}{l|}{WBB$(\boldsymbol{18},50)$ A} & 0. & 0. & 0. & 0. & 0. & 0. & 0. & \multicolumn{1}{c}{0.}& - & 0.404 & 0.408 & 0.405 & 0.406 & 0.404 & 0.403 & 0.405 & 0.407& - \\
\multicolumn{1}{l|}{} & \multicolumn{1}{l|}{WBB$(\boldsymbol{24},50)$ A} & 0. & 0. & 0. & 0. & 0. & 0. & 0. & \multicolumn{1}{c}{0.}& - & 0.394 & 0.406 & 0.398 & 0.401 & 0.394 & 0.393 & 0.394 & 0.404& - \\

\midrule

\multicolumn{1}{l|}{\multirow{5}{*}{\textit{$n$}}} & \multicolumn{1}{l|}{WBB$(4,\boldsymbol{5})$ A} & 0. & 0. & 0. & 0. & 0. & 0. & 0. & \multicolumn{1}{c}{0.} & -  & 0.495 & 0.495 & 0.499 & 0.493 & 0.490 & 0.493 & 0.498 & 0.495& - \\
\multicolumn{1}{l|}{} & \multicolumn{1}{l|}{WBB$(4,\boldsymbol{100})$ A} & 0. & 0. & 0. & 0. & 0. & 0. & 0. & \multicolumn{1}{c}{0.}& - & 0.406 & 0.410 & 0.408 & 0.408 & 0.409 & 0.412 & 0.406 & 0.410& - \\
    \multicolumn{1}{l|}{} & \multicolumn{1}{l|}{WBB$(4,\boldsymbol{150})$ A} & 0. & 0. & 0. & 0. & 0. & 0. & 0. & \multicolumn{1}{c}{0.}& - & 0.396 & 0.396 & 0.395 & 0.390 & 0.398 & 0.399 & 0.392 & 0.397& - \\
    \multicolumn{1}{l|}{} & \multicolumn{1}{l|}{WBB$(4,\boldsymbol{200})$ A} & 0. & 0. & 0. & 0. & 0. & 0. & 0. & \multicolumn{1}{c}{0.}& - & 0.396 & 0.387 & 0.390 & 0.387 & 0.390 & 0.388 & 0.389 & 0.382& - \\
    \multicolumn{1}{l|}{} & \multicolumn{1}{l|}{WBB$(4,\boldsymbol{250})$ A} & 0. & 0. & 0. & 0. & 0. & 0. & 0. & \multicolumn{1}{c}{0.}& - & 0.373 & 0.383 & 0.383 & 0.379 & 0.380 & 0.379 & 0.379 & 0.384& - \\

\bottomrule
\end{tabular}
}
\end{table*}

%% file: tables/recall.tex
\begin{table*}[tb]
\centering
\mycaption{Mean Recall is achieved by pooling the documents retrieved for each obfuscation query. The mechanism WBB$(k,n)$ has been evaluated by adopting as the similarity function the angle base distance (\textbf{A}).}
\resizebox{\textwidth}{!}{
\begin{tabular}{llcccccccccccccccccc}
\toprule
& & \multicolumn{9}{c}{\textbf{Robust '04}} & \multicolumn{9}{c}{\textbf{DL '19}} \\
\cmidrule(lr){3-11} \cmidrule(lr){12-20}
\textbf{IR model} & \textbf{Mechanism} & \multicolumn{9}{c}{\textbf{$\mathbf{\varepsilon}$}} & \multicolumn{9}{c}{\textbf{$\mathbf{\varepsilon}$}} \\
& & 1 & 5 & 10 &12.5& 15 &17.5& 20 & 50 & No-DP & 1 & 5 & 10 & 12.5 &15&17.5 & 20 & 50 & No-DP \\
\midrule
\multirow{9}{*}{\textbf{BM25}} 
    & No privacy & -     & -     & -     & -     & -     & -     & -   & -   & 0.410   & -   & -     & -     & -     & -     & -     & -   & -   & 0.454 \\
    \cline{2-20}
    & \aea & -     & -     & -     & -     & -     & -     & -   & -   & 0.420          & -    & -    & -    & -    & -    & -    & -  & -  & 0.445 \\
    & \fea& -     & -     & -     & -     & -     & -     & -   & -   & 0.140          & -    & -    & -    & -    & -    & -    & -  & -  & 0.231 \\
    & CMP     & 0.020 & 0.146 & 0.483 & 0.510 & 0.489 & 0.442 & 0.421 & 0.407 & -         & 0.011 & 0.135 & 0.384 & 0.534 & 0.517 & 0.514 & 0.480 & 0.444 & - \\
    & Mhl     & 0.032 & 0.089 & 0.398 & 0.500 & 0.512 & 0.501 & 0.466 & 0.407 & -         & 0.000 & 0.118 & 0.294 & 0.411 & 0.515 & 0.529 & 0.529 & 0.444 & -  \\
    \cline{2-20}
    & WBB (2,15) \textbf{A} & 0.119 & 0.131   & 0.141 & 0.143 & 0.141 & 0.147 & 0.139 & 0.145 & -    & 0.193 & 0.195 & 0.185 & 0.190 & 0.182 & 0.178 & 0.189 & 0.196 & -\\
    & WBB (2,20) \textbf{A} & 0.115 & 0.124   & 0.142 & 0.133 & 0.143 & 0.134 & 0.131 & 0.136 & -    & 0.187 & 0.195 & 0.176 & 0.186 & 0.188 & 0.177 & 0.188 & 0.193 & -\\
    & WBB (4,15) \textbf{A} & 0.119 & 0.131   & 0.137 & 0.138 & 0.139 & 0.132 & 0.141 & 0.144 & -    & 0.191 & 0.179 & 0.191 & 0.196 & 0.184 & 0.189 & 0.196 & 0.189 & -\\
    & WBB (4,20) \textbf{A} & 0.117   & 0.124 & 0.142 & 0.143 & 0.138 & 0.134 & 0.139 & 0.139 & -    & 0.172 & 0.180 & 0.187 & 0.179 & 0.179 & 0.175 & 0.178 & 0.185 & -\\
\midrule
\multirow{9}{*}{\textbf{TF-IDF}} 
    & No privacy & -     & -     & -     & -     & -     & -     & -   & -   & 0.411          & -    & -    & -    & -    & -    & -    & -  & -  & 0.451 \\
    \cline{2-20}
    & \aea& -     & -     & -     & -     & -     & -     & -   & -   & 0.420          & -    & -    & -    & -    & -    & -    & -  & -  & 0.443 \\
    & \fea& -     & -     & -     & -     & -     & -     & -   & -   & 0.139          & -    & -    & -    & -    & -    & -    & -  & -  & 0.231 \\
    & CMP     & 0.020 & 0.146 & 0.487 & 0.512 & 0.491 & 0.444 & 0.423 & 0.408 & -          & 0.011 & 0.135 & 0.386 & 0.535 & 0.515 & 0.515 & 0.478 & 0.442 & - \\
    & Mhl     & 0.032 & 0.089 & 0.398 & 0.504 & 0.516 & 0.504 & 0.468 & 0.408 & -          & 0.000 & 0.118 & 0.295 & 0.412 & 0.515 & 0.530 & 0.527 & 0.442 & - \\
    \cline{2-20}
    & WBB (2,15) \textbf{A} & 0.120 & 0.128 & 0.143 & 0.143 & 0.141 & 0.148 & 0.140 & 0.145 & -     & 0.190 & 0.195 & 0.184 & 0.190 & 0.181 & 0.178 & 0.188 & 0.194 & -\\
    & WBB (2,20) \textbf{A} & 0.114 & 0.126 & 0.140 & 0.136 & 0.143 & 0.134 & 0.132 & 0.137 & -     & 0.186 & 0.194 & 0.175 & 0.183 & 0.188 & 0.175 & 0.187 & 0.191 & -\\
    & WBB (4,15) \textbf{A} & 0.122 & 0.130 & 0.137 & 0.137 & 0.139 & 0.134 & 0.140 & 0.143 & -     & 0.190 & 0.180 & 0.190 & 0.194 & 0.183 & 0.190 & 0.197 & 0.188 & -\\
    & WBB (4,20) \textbf{A} & 0.118 & 0.124 & 0.142 & 0.144 & 0.137 & 0.134 & 0.136 & 0.136 & -     & 0.172 & 0.181 & 0.185 & 0.178 & 0.179 & 0.175 & 0.177 & 0.183 & -\\
\midrule
\multirow{9}{*}{\textbf{Contriever}} 
    & No privacy & -     & -     & -     & -     & -     & -     & -   & -   & 0.392   & -     & -     & -     & -     & -     & -     & -   & -   & 0.528 \\
    \cline{2-20}
    & \aea& -     & -     & -     & -     & -     & -     & -   & -   & 0.419          & -    & -    & -    & -    & -    & -    & -  & -  & 0.497 \\
    & \fea& -     & -     & -     & -     & -     & -     & -   & -   & 0.204          & -    & -    & -    & -    & -    & -    & -  & -  & 0.204 \\
    & CMP     & 0.034 & 0.106 & 0.469 & 0.507 & 0.481 & 0.433 & 0.406 & 0.392 & - & 0.000 & 0.077 & 0.446 & 0.615 & 0.644 & 0.628 & 0.576 & 0.512 & - \\
    & Mhl     & 0.026 & 0.088 & 0.345 & 0.473 & 0.503 & 0.497 & 0.460 & 0.392 & - & 0.000 & 0.057 & 0.264 & 0.476 & 0.601 & 0.641 & 0.650 & 0.512 & - \\
    \cline{2-20}
    & WBB (2,15) \textbf{A}     & 0.361   & 0.351   & 0.337   & 0.337   & 0.326   & 0.335   & 0.335 & 0.340 & -    & 0.347   & 0.341   & 0.352   & 0.349   & 0.347   & 0.352   & 0.345 & 0.351 & -\\
    & WBB (2,20) \textbf{A}    & 0.366   & 0.346   & 0.330   & 0.332   & 0.332   & 0.343   & 0.332 & 0.327 & -     & 0.380   & 0.386   & 0.384   & 0.397   & 0.394   & 0.393   & 0.393 & 0.380 & -\\
    & WBB (4,15) \textbf{A}    & 0.352   & 0.339   & 0.319   & 0.318   & 0.315   & 0.315   & 0.318 & 0.322 & -     & 0.342   & 0.352  & 0.365   & 0.351   & 0.346   & 0.348   & 0.348 & 0.362 & -\\
    & WBB (4,20) \textbf{A}    & 0.359   & 0.344   & 0.313   & 0.317   & 0.310   & 0.315   & 0.316 & 0.320 & -     & 0.240   & 0.270   & 0.364   & 0.270   & 0.238   & 0.289   & 0.289 & 0.258 & -\\
\midrule
\multirow{9}{*}{\textbf{Tas-B}} 
    & No privacy & -     & -     & -     & -     & -     & -     & -   & -   & 0.358   & -     & -     & -     & -     & -     & -     & -   & -   & 0.518 \\
    \cline{2-20}
    & \aea& -     & -     & -     & -     & -     & -     & -   & -   & 0.387          & -    & -    & -    & -    & -    & -    & -  & -  & 0.491 \\
    & \fea& -     & -     & -     & -     & -     & -     & -   & -   & 0.161          & -    & -    & -    & -    & -    & -    & -  & -  & 0.238 \\
    & CMP     & 0.027 & 0.080 & 0.434 & 0.477 & 0.444 & 0.398 & 0.369 & 0.356 & - & 0.000 & 0.063 & 0.392 & 0.615 & 0.636 & 0.622 & 0.575 & 0.498 & - \\
    & Mhl     & 0.028 & 0.064 & 0.310 & 0.438 & 0.464 & 0.460 & 0.423 & 0.356 & - & 0.000 & 0.042 & 0.275 & 0.455 & 0.584 & 0.645 & 0.638 & 0.499 & - \\
    \cline{2-20}
    & WBB (2,15) \textbf{A}     & 0.315   & 0.335   & 0.352   & 0.322   & 0.313   & 0.318   & 0.325 & 0.315 & -     & 0.322   & 0.321  & 0.323   & 0.325   & 0.325   & 0.333  & 0.332 & 0.338 & - \\
    & WBB (2,20) \textbf{A}    & 0.314   & 0.332   & 0.360   & 0.314   & 0.314   & 0.322   & 0.314 & 0.320 & -     & 0.347   & 0.356  & 0.356   & 0.349   & 0.353   & 0.359   & 0.359 & 0.352 & - \\
    & WBB (4,15) \textbf{A}    & 0.300   & 0.331   & 0.350   & 0.300   & 0.301   & 0.301   & 0.301 & 0.302 & -     & 0.312   & 0.318  & 0.327   & 0.323   & 0.328   & 0.327   & 0.326 & 0.324 & - \\
    & WBB (4,20) \textbf{A}    & 0.305   & 0.331   & 0.350   & 0.301   & 0.302   & 0.301   & 0.300 & 0.300 & -     & 0.218  & 0.244  & 0.334   & 0.239   & 0.231   & 0.220   & 0.264 & 0.215 & - \\
\bottomrule
\end{tabular}}
\label{tab:recall_results}
\end{table*}

%% file: figures/cm_combined_R.tex
\begin{figure*}[h!]
        \centering
        \includegraphics[width=\linewidth]{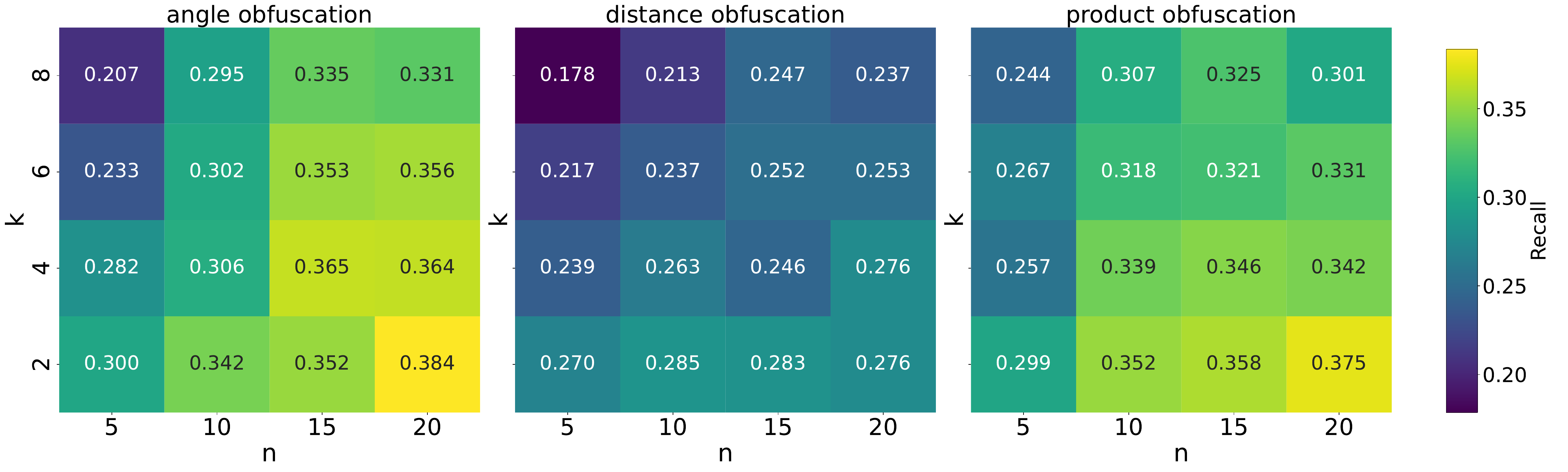} 
        \mycaption{Comparison matrices of the Recall for different obfuscations (angle, euclidean distance, and product), each at the same level of $\varepsilon$-\ac{DP} ($\varepsilon=10$), computed using Contriever and the DL'19 collection. With no privacy, Recall$=0.528$}
        \label{fig:confmatrix_Recall}
\end{figure*}

%% file: tables/nDCG_NEW.tex
\begin{table*}[tb]
\centering
\caption{Mean nDCG@10 is achieved by pooling the documents reranked for each obfuscation query. The mechanism WBB$(k,n)$ has been evaluated by adopting as the similarity function the angle base distance (\textbf{A}).}
\resizebox{\textwidth}{!}{
\begin{tabular}{llcccccccccccccccccc}
\toprule
& & \multicolumn{9}{c}{\textbf{Robust '04}} & \multicolumn{9}{c}{\textbf{DL '19}} \\
\cmidrule(lr){3-11} \cmidrule(lr){12-20}
\textbf{IR model} & \textbf{Mechanism} & \multicolumn{9}{c}{\textbf{$\mathbf{\varepsilon}$}} & \multicolumn{9}{c}{\textbf{$\mathbf{\varepsilon}$}} \\
& & 1 & 5 & 10 &12.5& 15 &17.5& 20 & 50 & No-DP & 1 & 5 & 10 & 12.5 &15&17.5 & 20 & 50 & No-DP \\
\midrule
\multirow{9}{*}{\textbf{BM25}} 
    & No privacy & -     & -     & -     & -     & -     & -     & -   & -   & 0.477   & -   & -     & -     & -     & -     & -     & -   & -   & 0.675 \\
    \cline{2-20}
    & \aea& -     & -     & -     & -     & -     & -     & -   & -   & 0.423        & -    & -    & -    & -    & -    & -    & -  & -  & 0.557 \\
    & \fea& -     & -     & -     & -     & -     & -     & -   & -   & 0.147        & -    & -    & -    & -    & -    & -    & -  & -  & 0.069 \\
    & CMP     & 0. & 0.002 & 0.091 & 0.226 & 0.338 & 0.402 & 0.415 & 0.421 & - & 0. & 0. & 0.053 & 0.164 & 0.293 & 0.363 & 0.394 & 0.403 & - \\
    & Mhl     & 0. & 0.001 & 0.037 & 0.106 & 0.204 & 0.303 & 0.369 & 0.421 & - & 0. & 0. & 0.016 & 0.067 & 0.154 & 0.253 & 0.329 & 0.403 & - \\
    \cline{2-20}
    & WBB (2,15) \textbf{A} & 0.099 & 0.106 & 0.115 & 0.111 & 0.125 & 0.116 & 0.105 & 0.111 & -    & 0.225 & 0.226 & 0.225 & 0.224 & 0.236 & 0.223 & 0.236 & 0.248 & -\\
    & WBB (2,20) \textbf{A} & 0.092 & 0.103 & 0.117 & 0.110 & 0.111 & 0.119 & 0.109 & 0.116 & -    & 0.230 & 0.230 & 0.215 & 0.222 & 0.243 & 0.203 & 0.241 & 0.236 & -\\
    & WBB (4,15) \textbf{A} & 0.092 & 0.106 & 0.104 & 0.100 & 0.117 & 0.112 & 0.119 & 0.111 & -    & 0.225 & 0.201 & 0.239 & 0.215 & 0.221 & 0.209 & 0.225 & 0.213 & -\\
    & WBB (4,20) \textbf{A} & 0.096 & 0.095 & 0.115 & 0.108 & 0.113 & 0.110 & 0.113 & 0.107 & -    & 0.189 & 0.198 & 0.209 & 0.198 & 0.202 & 0.186 & 0.191 & 0.221 & -\\
\midrule
\multirow{9}{*}{\textbf{TF-IDF}} 
    & No privacy & -     & -     & -     & -     & -     & -     & -   & -   & 0.477          & -    & -    & -    & -    & -    & -    & -  & -  & 0.675 \\
    \cline{2-20}
    & \aea& -     & -     & -     & -     & -     & -     & -   & -   & 0.421          & -    & -    & -    & -    & -    & -    & -  & -  & 0.557 \\
    & \fea& -     & -     & -     & -     & -     & -     & -   & -   & 0.148          & -    & -    & -    & -    & -    & -    & -  & -  & 0.068 \\
    & CMP     & 0. & 0.002 & 0.091 & 0.227 & 0.339 & 0.403 & 0.417 & 0.423 & - & 0. & 0. & 0.053 & 0.164 & 0.292 & 0.361 & 0.393 & 0.402 & - \\
    & Mhl     & 0. & 0.001 & 0.040 & 0.106 & 0.205 & 0.304 & 0.370 & 0.423 & - & 0. & 0. & 0.016 & 0.067 & 0.154 & 0.252 & 0.329 & 0.402 & - \\
    \cline{2-20}
    & WBB (2,15) \textbf{A} & 0.099 & 0.102 & 0.113 & 0.110 & 0.123 & 0.117 & 0.106 & 0.110 & -    & 0.224 & 0.227 & 0.223 & 0.223 & 0.236 & 0.223 & 0.235 & 0.249 & -\\
    & WBB (2,20) \textbf{A} & 0.093 & 0.103 & 0.117 & 0.108 & 0.112 & 0.119 & 0.113 & 0.114 & -    & 0.232 & 0.228 & 0.216 & 0.223 & 0.244 & 0.201 & 0.241 & 0.236 & -\\
    & WBB (4,15) \textbf{A} & 0.094 & 0.102 & 0.106 & 0.100 & 0.117 & 0.115 & 0.116 & 0.111 & -    & 0.223 & 0.200 & 0.238 & 0.215 & 0.223 & 0.209 & 0.225 & 0.213 & -\\
    & WBB (4,20) \textbf{A} & 0.097 & 0.095 & 0.114 & 0.109 & 0.114 & 0.111 & 0.114 & 0.106 & -    & 0.190 & 0.200 & 0.207 & 0.197 & 0.202 & 0.188 & 0.191 & 0.222 & -\\
\midrule
\multirow{9}{*}{\textbf{Contriever}} 
    & No privacy & -     & -     & -     & -     & -     & -     & -   & -   & 0.466   & -     & -     & -     & -     & -     & -     & -   & -   & 0.676 \\
    \cline{2-20}
    & \aea& -     & -     & -     & -     & -     & -     & -   & -   & 0.430          & -    & -    & -    & -    & -    & -    & -  & -  & 0.567 \\
    & \fea& -     & -     & -     & -     & -     & -     & -   & -   & 0.200          & -    & -    & -    & -    & -    & -    & -  & -  & 0.056 \\
    & CMP     & 0. & 0.002 & 0.093 & 0.239 & 0.373 & 0.440 & 0.459 & 0.466 & - & 0. & 0. & 0.041 & 0.189 & 0.397 & 0.522 & 0.572 & 0.598 & - \\
    & Mhl     & 0. & 0.001 & 0.036 & 0.104 & 0.211 & 0.323 & 0.402 & 0.466 & - & 0. & 0. & 0.014 & 0.064 & 0.178 & 0.330 & 0.470 & 0.597 & - \\
    \cline{2-20}
    & WBB (2,15) \textbf{A}     & 0.457   & 0.454   & 0.447   & 0.445   & 0.444   & 0.452   & 0.442 & 0.448 & -    & 0.564   & 0.577   & 0.599   & 0.582   & 0.607   & 0.584   & 0.576 & 0.600 & -\\
    & WBB (2,20) \textbf{A}    & 0.460   & 0.451   & 0.444   & 0.444   & 0.449   & 0.450   & 0.445 & 0.444 & -     & 0.597   & 0.623   & 0.594   & 0.613   & 0.604   & 0.626   & 0.623 & 0.603 & -\\
    & WBB (4,15) \textbf{A}    & 0.450   & 0.447   & 0.431   & 0.435   & 0.432   & 0.432   & 0.433 & 0.442 & -     & 0.599   & 0.616  & 0.611   & 0.588   & 0.606   & 0.612   & 0.589 & 0.611 & -\\
    & WBB (4,20) \textbf{A}    & 0.457   & 0.447   & 0.430   & 0.435   & 0.423   & 0.439   & 0.433 & 0.432 & -     & 0.453   & 0.502   & 0.584   & 0.527   & 0.493   & 0.467   & 0.493 & 0.497 & -\\
\midrule
\multirow{9}{*}{\textbf{Tas-B}} 
    & No privacy & -     & -     & -     & -     & -     & -     & -   & -   & 0.469   & -     & -     & -     & -     & -     & -     & -   & -   & 0.674 \\
    \cline{2-20}
    & \aea& -     & -     & -     & -     & -     & -     & -   & -   & 0.431          & -    & -    & -    & -    & -    & -    & -  & -  & 0.579 \\
    & \fea& -     & -     & -     & -     & -     & -     & -   & -   & 0.168          & -    & -    & -    & -    & -    & -    & -  & -  & 0.052 \\
    & CMP     & 0. & 0.001 & 0.085 & 0.225 & 0.355 & 0.423 & 0.439 & 0.446 & - & 0. & 0. & 0.048 & 0.195 & 0.415 & 0.552 & 0.615 & 0.649 & - \\
    & Mhl     & 0. & 0.001 & 0.032 & 0.097 & 0.196 & 0.305 & 0.384 & 0.446 & - & 0. & 0. & 0.015 & 0.069 & 0.180 & 0.340 & 0.488 & 0.648 & - \\
    \cline{2-20}
    & WBB (2,15) \textbf{A}     & 0.457   & 0.448   & 0.440   & 0.443   & 0.441   & 0.445   & 0.442 & 0.442 & -     & 0.573   & 0.580  & 0.590   & 0.623   & 0.585   & 0.598  & 0.589 & 0.595 & - \\
    & WBB (2,20) \textbf{A}    & 0.460   & 0.449   & 0.447   & 0.446   & 0.444   & 0.446   & 0.443 & 0.442 & -     & 0.622   & 0.618  & 0.622   & 0.625   & 0.616   & 0.630   & 0.622 & 0.621 & - \\
    & WBB (4,15) \textbf{A}    & 0.454   & 0.451   & 0.436   & 0.438   & 0.432   & 0.431   & 0.430 & 0.442 & -     & 0.616   & 0.593  & 0.606   & 0.607   & 0.608   & 0.604   & 0.609 & 0.594 & - \\
    & WBB (4,20) \textbf{A}    & 0.454   & 0.443   & 0.428   & 0.436   & 0.427   & 0.427   & 0.425 & 0.429 & -     & 0.467  & 0.513  & 0.613   & 0.522   & 0.486   & 0.482   & 0.460 & 0.496 & - \\
\bottomrule
\end{tabular}}
\label{tab:nDCG_results}
\end{table*}

%% file: figures/cm_combined_nDCG_at_10.tex
\begin{figure*}[h!]
        \centering
        \includegraphics[width=\linewidth]{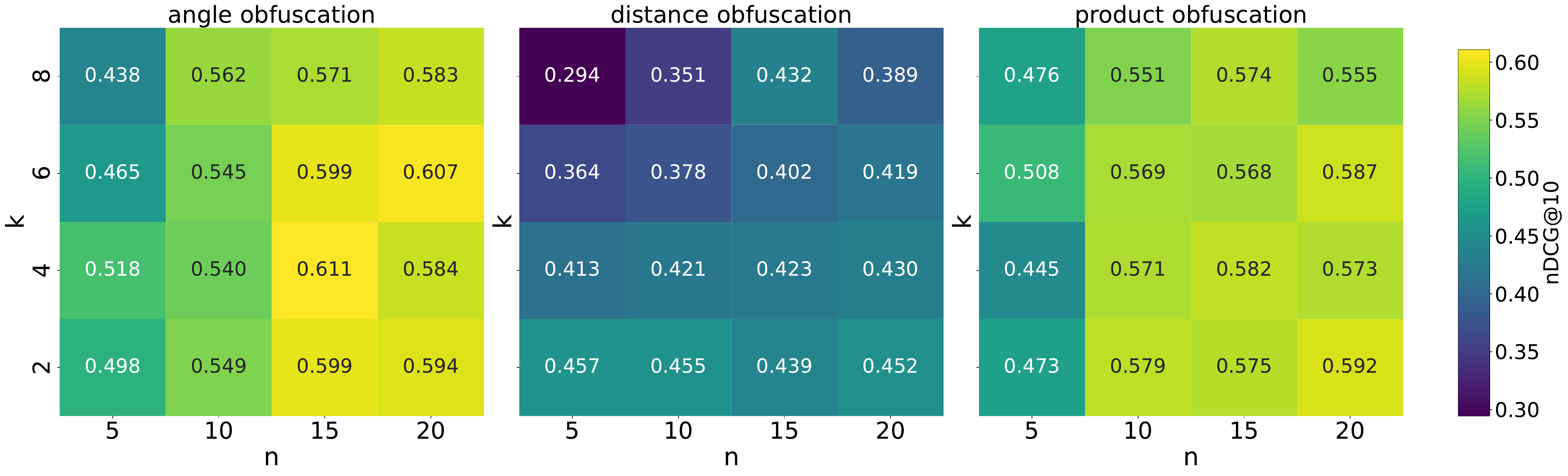} 
        \mycaption{Comparison matrices of nDCG@10 for different obfuscations (angle, euclidean distance, and product), each at the same level of $\varepsilon$-\ac{DP} ($\varepsilon=10$), computing using Contriever and the DL'19 collection. With no privacy, nDCG@10$=0.676$.}
        \label{fig:confmatrix}
\end{figure*}

%% file: sections/05relatedWorks.tex
\section{Related Works}
\label{sec:relatedWorks}
\paragraph{\ac{IR} and Privacy}
The problem of Privacy is not new to \ac{IR}~\citep{domingoferrer2008surveyprivacydatamining,arampatzis2013versatiletoolprivacywebsearch,wasi2018queryobfuscationforprivacyinpersonalizedwebsearch}. To address such a problem, the \ac{IR} community has approached the problem of obfuscating queries in three different ways: i) methods based on dummy queries; ii) methods based on unlinkability; iii) methods based on proxy queries. One aspect of the dummy queries approach~\citep{elovici2002privacymodelforwebsurf,DomingoFerrer2009HkprivateIR} involves sending, alongside the original query to the \ac{IRS}, a collection of unrelated queries that are syntactically similar to the user's one. On the other hand, the unlinkability approach~\citep{domingoferrer2009PIRinpeer2peer,castellaroca2009userprivacyinse} depends on cryptographic and Private Information Retrieval methods to enable a community of users to share queries among themselves, ensuring that each user submits the query of others and the system cannot create a specific user profile. The approach based on proxy queries~\citep{frobe2022efficentqueryobf,arampatzis2015versativescramblerprivatewebsearch} entails the decomposition of the original query into several non-sensitive queries, the combination of which may potentially yield an answer to the user's information need, e.g., the query ``throat cancer" is transformed into ``neck" and ``tumour", reducing the information disclosure of each query. The drawback of dummy and unlinkability approaches is that the original query is sent to the \ac{IRS}, hence increasing the vulnerability toward machine learning-based attacks~\citep{peddinti2011anonymizingnetworksforwebse, peddinti2014websearchqueryprivacy}. In addition, the unlinkability approach poses another potential threat to the user: it moves the problem from the \ac{IRS} to a different user in the federation. These risks do not occur for the approach based on the proxy queries. However, dummy queries physiologically decrease the effectiveness of the retrieval, which does not occur for the dummy and unlinkability query approaches. Regardless, we argue that the user willingly renounces part of the retrieval capability in favour of achieving strong privacy protection. Therefore, we focus on analyzing obfuscation techniques that rely on proxy queries and exploring various methods, including those based on \ac{DP} and other approaches.

%% file: sections/06conclusionsAndFutureWorks.tex
\section{Conclusions and Future Works}
\label{sec:conclusionsAndFutureWorks}
In this study, we have presented \acf{WBB}, a new method based on $\varepsilon$-\ac{DP} for query obfuscation in \ac{IR} task. Our approach involves removing the top-$k$ most similar words, depending on a scoring function $m$, and sampling the substitutes from a list of potential $n$ candidates to disguise other possible relevant terms employing the exponential mechanism. In this study, we evaluated various mechanisms in the field of \ac{IR} by analyzing the Recall and nDCG using two TREC Collections and multiple retrieval models, sparse bag-of-word, i.e., BM25 and TF-IDF, and dense bi-encoders, i.e., Contriever and TAS-B. We compared \ac{WBB} with other query obfuscation mechanisms, the ones proposed by \citet{arampatzis2011enenhcingdeniabilityquerylogs} and \citet{frobe2022efficentqueryobf}, and other mechanisms designed for text obfuscation in \ac{NLP} based on \ac{DP}, namely the \acf{CMP}~\citep{feyisetan2020multivariate} and the \acf{MHL}~\citep{xu2020mahalanobis}. To assess the trade-off between privacy and utility of the $20$ obfuscated queries generated by the examined mechanisms, we considered different $\varepsilon$ parametrizations. Our findings highlight that the \ac{WBB} mechanism achieves an average nDCG score of 0.622 on the DeepLearning'19 collection, approaching the score obtained using the original queries. In addition, our study delved into the influence between the parametrization of the \ac{WBB} mechanism and the values of nDCG. Moreover, the privacy analysis showed that the \ac{WBB} overcomes the limitation of only formal privacy, ensuring that, even for high privacy budgets $\varepsilon$, the Lexical and Semantic Similarity of the obfuscated queries are always below the ones provided by State-of-the-Art mechanisms. In future directions, we plan to explore how to construct obfuscated queries by studying new scoring functions $m$ and the best parameter selection of the values that identify the top-$k$ and the $n$ candidates. Furthermore, we intend to focus the study on analyzing ways to enhance the effectiveness of obfuscated queries by exploring the perturbation of contextualized dense representations. Finally, we plan to explore potential inference attacks against the \ac{WBB} mechanism. 

%% file: references.tex
\bibliographystyle{elsarticle-harv}
\bibliography{references}